\documentclass[a4paper,english]{lipics}
\usepackage[utf8]{inputenc}
\usepackage{amsmath}
\usepackage{amssymb}
\usepackage[matrix,arrow]{xy}
\CompileMatrices
\usepackage{wrapfig}
\usepackage{tikz}
\usepackage{xypic}
\usepackage{strid}

\newcommand{\N}{\mathbb{N}}
\newcommand{\R}{\mathbb{R}}
\newcommand{\Rplus}{\mathbb{R}^+}

\newcommand{\intset}[1]{\langle #1\rangle}
\newcommand{\ns}[1]{{\null^\ast #1}}
\newcommand{\nsint}{\ns\N}
\newcommand{\nsreal}{\ns\R}

\newcommand{\nsclass}[1]{\langle #1 \rangle}

\newcommand{\ie}{i.e.~}
\newcommand{\eg}{e.g.~}

\newcommand{\MARK}[1]{\marginpar{\tiny Things changed here!}}

\newcommand{\qeq}{\quad=\quad}

\newcommand{\qmapsto}{\quad\mapsto\quad}

\newcommand{\resp}{resp.~}
\newcommand{\wrt}{wrt~}
\newcommand{\category}[1]{\mathbf{#1}}

\newcommand{\C}{\mathcal{C}}

\newcommand{\F}{\mathcal{F}}
\renewcommand{\st}{\mathop{\mathrm{st}}}
\newcommand{\set}[1]{\{#1\}}
\newcommand{\setof}[2]{\set{#1\mid #2}}
\newcommand{\Tr}[2]{\mathop{\mathrm{Tr}_{#1}^{#2}}}
\newcommand{\Net}{\category{Net}}
\newcommand{\sNet}{\category{sNet}}
\newcommand{\id}{\mathrm{id}}
\newcommand{\nbd}{\nobreakdash-\hspace{0pt}}
\newcommand{\Mod}{\category{Mod}}
\newcommand{\Fix}{\category{Fix}}
\renewcommand{\leq}{\leqslant}

\newcommand{\pa}[1]{\left(#1\right)}
\newcommand{\fix}[2][]{\mathop{\mathrm{fix}_{#1}}(#2)}
\newcommand{\Cpo}{\category{Cpo}}
\newcommand{\ICpo}{\category{ICpo}}

\newcommand{\sdom}[2]{#2^{\leq #1}}

\newcommand{\CT}{CT}
\newcommand{\CCT}{CCT}
\newcommand{\IT}{IT}
\newcommand{\floor}[1]{\lfloor #1\rfloor}

\theoremstyle{plain}
\newtheorem{proposition}{Proposition}

\title{A Non-Standard Semantics for Kahn Networks in Continuous
  Time\footnote{This work was partially supported by the Office of Naval
    Research and by the PANDA (``Parallel and Distributed Analysis'',
    \hbox{ANR-09-BLAN-0169}) French ANR project.}}

\author[1]{Romain Beauxis}
\author[2]{Samuel Mimram}
\affil[1]{Tulane University, Department of Mathematics\\
6823 St Charles Avenue, New Orleans LA 70118, USA\\\texttt{rbeauxis@tulane.edu}}
\affil[2]{CEA, LIST\\Point Courrier 94, 91191 Gif-sur-Yvette, France\\\texttt{samuel.mimram@cea.fr}}
\Copyright[nc-nd]{Romain Beauxis and Samuel Mimram}

\begin{document}
\maketitle

\begin{abstract}
  In a seminal article, Kahn has introduced the notion of process network and
  given a semantics for those using Scott domains whose elements are (possibly
  infinite) sequences of values. This model has since then become a standard
  tool for studying distributed asynchronous computations. From the beginning,
  process networks have been drawn as particular graphs, but this syntax is
  never formalized. We take the opportunity to clarify it by giving a precise
  definition of these graphs, that we call \emph{nets}.  The resulting category
  is shown to be a \emph{fixpoint category}, \ie a cartesian category which is
  traced \wrt the monoidal structure given by the product, and interestingly
  this structure characterizes the category: we show that it is the free
  fixpoint category containing a given set of morphisms, thus providing a
  complete axiomatics that models of process networks should satisfy. We then
  use these tools to build a model of networks in which data vary over a
  \emph{continuous} time, in order to elaborate on the idea that process
  networks should also be able to encompass computational models such as
  hybrid systems or electric circuits. We relate this model to Kahn's semantics
  by introducing a third model of networks based on non-standard analysis, whose
  elements form an \emph{internal complete partial order} for which many
  properties of standard domains can be reformulated. The use of hyperreals in
  this model allows it to formally consider the notion of infinitesimal, and
  thus to make a bridge between discrete and continuous time: time is
  ``discrete'', but the duration between two instants is infinitesimal. Finally,
  we give some examples of uses of the model by describing some networks
  implementing common constructions in analysis.
\end{abstract}

\vspace{3ex}
Process networks have been introduced by Kahn, together with an associated
semantics based on Scott domains, as one of the first model for concurrent and
asynchronous computations~\cite{kahn:semantics-parallel}. These networks are
constituted of \emph{processes} (which may be thought as computers on a network
or threads on a computer for instance) which perform computations and can
exchange information through \emph{channels} acting as unbounded FIFO
queues. Finite or infinite sequences of values, that are called \emph{streams},
are communicated on the channels, and the semantics of the whole process network
is considered to be the streams that can be exchanged, depending on the data
provided by the environment. The set of streams can be structured as a complete
partial order, and the semantics of networks is modeled by Scott-continuous
functions on this domain: the fact that these functions admit a smallest fixpoint
turns out to be crucial in order to model ``loops'' formed by channels
in the network. A series of subsequent works have provided a precise understanding of
this fixpoint construction~\cite{faustini1982operational, hildebrandt2004relational}.

In this model, time is \emph{discrete} in the sense that a countable number of
values might be exchanged during an execution: we can consider that each value
occurs at a given instant~$t\in\N$. In this article, we are interested in
understanding how to extend the usual semantics of process networks in order to
consider computations in \emph{continuous} time, by replacing~$\N$ by~$\R^+$ for
the domain of time, so as to embrace computational models such as electric
circuits or hybrid systems, with which it bears many similarities. However, how
would such a semantics relate to the usual discrete semantics of networks? The
fundamental intuition in order to relate the two models is the following: if we
allow ourselves to consider \emph{infinitesimal} durations~$\d t$, then we can
think of continuous time as being somehow ``discrete'', its instants being~$0,\d
t, 2\d t, 3\d t, \ldots$ This idea of \emph{infinitesimal time} originates in
the works of Bliudze and Krob~\cite{bliudze2009modelling}, and was later on
developed by Benveniste, Caillaud and Pouzet~\cite{benveniste2010fundamentals},
who formalized it by using tools provided by \emph{non-standard
  analysis}~\cite{robinson-book, goldblatt_lectureshyperreals_1998} in order to
give a rigorous meaning to infinitesimals.
Here, we develop on these ideas by structuring the resulting notion of stream
into \emph{internal Scott domains}, which are shown to provide a model of
process networks, and explain how the resulting model provides a nice bridge
between discrete and continuous time.

For this purpose, we introduce a new model for Kahn process networks. However,
what is precisely the syntax for these networks that we want to model? Here, we
formalize the definition of the graphs which are often used to informally
represent process networks, by a structure that we call \emph{nets}. We show
here that the resulting category is a \emph{fixpoint category}, \ie a cartesian
category which is traced~\cite{joyal-street-verity:traced-monoidal-categories}
\wrt the monoidal structure provided by products. Moreover, this structure
characterizes the category in the sense that the category of nets is a free
category of such kind.  This result thus provides a complete description of the
axioms that a model of nets should satisfy. We finally use this structure to
show that streams in infinitesimal time form such a model. We elaborate here on
a series of earlier works. The structure of traced monoidal category for the
Kahn model has been introduced by Hildebrandt, Panangaden and
Winskel~\cite{hildebrandt2004relational} and the construction of nets introduced
here is a generalization of the one introduced
in~\cite{hasegawa2008finite}. Properties of fixpoint categories and their
relationship to fixpoint operators have been studied in details by
Hasegawa~\cite{hasegawa1997recursion}.

We begin by defining nets (Section~\ref{sec:def-nets}), show that they are free
fixpoint categories (Section~\ref{sec:fixpoint-cat}) and explain that
Scott-domain semantics can be given for nets (Section~\ref{sec:net-models}). We
then recall basic constructions and properties of non-standard analysis
(Section~\ref{sec:hyperreals}), define the notion of internal domain of which
infinitesimal-time streams are an instance (Section~\ref{sec:internal-dom}) and
relate models in infinitesimal and continuous time (Section~\ref{sec:ct-it}). We
finally conclude in Section~\ref{sec:concl}.

\section{Nets and their semantics}
\label{sec:nets}
\begin{wrapfigure}{r}{2.7cm}
  \vspace{-4ex}
  \begin{tikzpicture}[baseline=(current bounding box.center),xscale=0.50,yscale=0.50,scale=0.80]
\useasboundingbox (-0.5,-0.5) rectangle (6.5,4.5);
\draw[] (0.00,4.00) -- (0.05,4.01) -- (0.11,4.01) -- (0.16,4.02) -- (0.21,4.03) -- (0.26,4.03) -- (0.32,4.04) -- (0.37,4.04) -- (0.42,4.05) -- (0.47,4.05) -- (0.52,4.05) -- (0.57,4.06) -- (0.62,4.06) -- (0.67,4.05) -- (0.72,4.05) -- (0.77,4.05) -- (0.82,4.04) -- (0.86,4.03) -- (0.91,4.03) -- (0.96,4.01) -- (1.00,4.00);
\draw[] (1.00,4.00) -- (1.06,3.98) -- (1.12,3.95) -- (1.18,3.92) -- (1.23,3.88) -- (1.29,3.85) -- (1.34,3.80) -- (1.40,3.76) -- (1.45,3.71) -- (1.50,3.66) -- (1.55,3.61) -- (1.59,3.56) -- (1.64,3.50) -- (1.69,3.44) -- (1.73,3.38) -- (1.78,3.32) -- (1.82,3.26) -- (1.87,3.19) -- (1.91,3.13) -- (1.96,3.06) -- (2.00,3.00);
\draw[] (2.00,3.00) -- (2.05,3.01) -- (2.11,3.01) -- (2.16,3.02) -- (2.21,3.03) -- (2.26,3.03) -- (2.32,3.04) -- (2.37,3.04) -- (2.42,3.05) -- (2.47,3.05) -- (2.52,3.05) -- (2.57,3.06) -- (2.62,3.06) -- (2.67,3.05) -- (2.72,3.05) -- (2.77,3.05) -- (2.82,3.04) -- (2.86,3.03) -- (2.91,3.03) -- (2.96,3.01) -- (3.00,3.00);
\draw[] (3.00,3.00) -- (3.06,2.98) -- (3.12,2.95) -- (3.18,2.92) -- (3.23,2.88) -- (3.29,2.85) -- (3.34,2.80) -- (3.40,2.76) -- (3.45,2.71) -- (3.50,2.66) -- (3.55,2.61) -- (3.59,2.56) -- (3.64,2.50) -- (3.69,2.44) -- (3.73,2.38) -- (3.78,2.32) -- (3.82,2.26) -- (3.87,2.19) -- (3.91,2.13) -- (3.96,2.06) -- (4.00,2.00);
\draw[] (4.00,2.00) -- (4.07,1.95) -- (4.15,1.91) -- (4.22,1.86) -- (4.29,1.81) -- (4.36,1.77) -- (4.43,1.72) -- (4.50,1.67) -- (4.56,1.62) -- (4.62,1.58) -- (4.68,1.53) -- (4.73,1.48) -- (4.78,1.43) -- (4.83,1.38) -- (4.87,1.33) -- (4.91,1.27) -- (4.94,1.22) -- (4.96,1.17) -- (4.98,1.11) -- (4.99,1.06) -- (5.00,1.00);
\draw[] (5.00,1.00) -- (5.00,0.94) -- (4.99,0.88) -- (4.97,0.82) -- (4.95,0.77) -- (4.93,0.71) -- (4.89,0.65) -- (4.85,0.59) -- (4.81,0.53) -- (4.76,0.47) -- (4.71,0.42) -- (4.65,0.37) -- (4.59,0.31) -- (4.53,0.26) -- (4.46,0.22) -- (4.39,0.17) -- (4.31,0.13) -- (4.24,0.09) -- (4.16,0.06) -- (4.08,0.03) -- (4.00,0.00);
\draw[] (4.00,0.00) -- (3.79,-0.05) -- (3.57,-0.07) -- (3.35,-0.07) -- (3.13,-0.05) -- (2.91,-0.00) -- (2.69,0.06) -- (2.47,0.14) -- (2.26,0.24) -- (2.06,0.35) -- (1.87,0.48) -- (1.70,0.61) -- (1.53,0.75) -- (1.39,0.90) -- (1.26,1.06) -- (1.15,1.22) -- (1.07,1.38) -- (1.01,1.54) -- (0.98,1.70) -- (0.97,1.85) -- (1.00,2.00);
\draw[] (1.00,2.00) -- (1.02,2.06) -- (1.04,2.11) -- (1.07,2.17) -- (1.10,2.22) -- (1.14,2.27) -- (1.18,2.33) -- (1.22,2.38) -- (1.27,2.43) -- (1.32,2.48) -- (1.37,2.53) -- (1.43,2.58) -- (1.49,2.62) -- (1.55,2.67) -- (1.61,2.72) -- (1.67,2.77) -- (1.74,2.81) -- (1.80,2.86) -- (1.87,2.91) -- (1.93,2.95) -- (2.00,3.00);
\draw[] (6.00,3.00) -- (5.95,3.01) -- (5.89,3.01) -- (5.84,3.02) -- (5.79,3.03) -- (5.74,3.03) -- (5.68,3.04) -- (5.63,3.04) -- (5.58,3.05) -- (5.53,3.05) -- (5.48,3.05) -- (5.43,3.06) -- (5.38,3.06) -- (5.33,3.05) -- (5.28,3.05) -- (5.23,3.05) -- (5.18,3.04) -- (5.14,3.03) -- (5.09,3.03) -- (5.04,3.01) -- (5.00,3.00);
\draw[] (5.00,3.00) -- (4.94,2.98) -- (4.88,2.95) -- (4.82,2.92) -- (4.77,2.88) -- (4.71,2.85) -- (4.66,2.80) -- (4.60,2.76) -- (4.55,2.71) -- (4.50,2.66) -- (4.45,2.61) -- (4.41,2.56) -- (4.36,2.50) -- (4.31,2.44) -- (4.27,2.38) -- (4.22,2.32) -- (4.18,2.26) -- (4.13,2.19) -- (4.09,2.13) -- (4.04,2.06) -- (4.00,2.00);
\draw[] (0.00,1.00) -- (0.16,0.96) -- (0.33,0.93) -- (0.49,0.89) -- (0.66,0.86) -- (0.82,0.83) -- (0.98,0.80) -- (1.14,0.78) -- (1.30,0.76) -- (1.46,0.74) -- (1.61,0.73) -- (1.76,0.72) -- (1.91,0.72) -- (2.06,0.73) -- (2.21,0.74) -- (2.35,0.76) -- (2.49,0.79) -- (2.62,0.83) -- (2.75,0.88) -- (2.88,0.93) -- (3.00,1.00);
\draw[] (3.00,1.00) -- (3.06,1.04) -- (3.11,1.07) -- (3.17,1.11) -- (3.22,1.15) -- (3.27,1.20) -- (3.32,1.24) -- (3.37,1.29) -- (3.43,1.34) -- (3.48,1.39) -- (3.52,1.44) -- (3.57,1.49) -- (3.62,1.55) -- (3.67,1.60) -- (3.72,1.66) -- (3.77,1.71) -- (3.81,1.77) -- (3.86,1.83) -- (3.91,1.88) -- (3.95,1.94) -- (4.00,2.00);
\draw[] (6.00,2.00) -- (5.95,2.01) -- (5.89,2.01) -- (5.84,2.02) -- (5.79,2.03) -- (5.74,2.03) -- (5.68,2.04) -- (5.63,2.04) -- (5.58,2.05) -- (5.53,2.05) -- (5.48,2.05) -- (5.43,2.06) -- (5.38,2.06) -- (5.33,2.05) -- (5.28,2.05) -- (5.23,2.05) -- (5.18,2.04) -- (5.14,2.03) -- (5.09,2.03) -- (5.04,2.01) -- (5.00,2.00);
\draw[] (5.00,2.00) -- (4.94,1.98) -- (4.88,1.95) -- (4.82,1.92) -- (4.77,1.88) -- (4.71,1.85) -- (4.66,1.80) -- (4.60,1.76) -- (4.55,1.71) -- (4.50,1.66) -- (4.45,1.61) -- (4.41,1.56) -- (4.36,1.50) -- (4.31,1.44) -- (4.27,1.38) -- (4.22,1.32) -- (4.18,1.26) -- (4.13,1.19) -- (4.09,1.13) -- (4.04,1.06) -- (4.00,1.00);
\filldraw[fill=white] (1.50, 2.00)  -- (1.50,4.00) -- (2.50,4.00) -- (2.50,2.00) -- (1.50,2.00);
\filldraw[fill=white] (3.50, 1.00)  -- (3.50,3.00) -- (4.50,3.00) -- (4.50,1.00) -- (3.50,1.00);
\draw (2.00,3.00) node{$\alpha$};
\draw (4.00,2.00) node{$\beta$};
\end{tikzpicture}
  \vspace{-5ex}
\end{wrapfigure}
A Kahn process network~\cite{kahn:semantics-parallel} can be thought as a finite
set of boxes, with inputs and outputs, linked through wires, producing outputs
depending on their inputs which will be asynchronously transmitted over the
wires. Over time, data circulates through the network, producing streams of
data.
The dataflow semantics of these networks has been well-studied in relation with
Scott domains and category theory~\cite{panangaden1990stability,
  hildebrandt2004relational}. However, there is no canonical syntax for them,
even though a graphical notation (as pictured on the right) is often used. Since
the purpose of this paper is to provide a new semantics for process networks, we
take this opportunity to clarify the definition of the syntax, by formalizing
the graphical notation and relating it with the categorical structure of the
models. The ideas developed here in order to develop an axiomatics for Kahn
process networks originate in various previous works in the field. Kahn's
original paper~\cite{kahn:semantics-parallel} mentions results of
decidability of the equivalence of graphs representing networks (which are
called \emph{schemata})
based on earlier works~\cite{courcelle1974algorithmes}. Many subsequent articles
underline the importance of operations on networks such as sequential
composition, parallel composition, tupling (products) and
feedback~\cite{faustini1982operational,panangaden1990stability}, and a traced
monoidal category modeling Kahn networks was constructed
in~\cite{hildebrandt2004relational}. On the categorical side, the ``drawings''
used here have been formalized as \emph{string diagrams} representing morphisms
in monoidal categories~\cite{joyal1991geometry}. Traced monoidal categories were
introduced in~\cite{joyal-street-verity:traced-monoidal-categories} and turned
out to be a fundamental tool in computer
science~\cite{abramsky1996retracing}. Their axiomatics was simplified and
studied in the cartesian case~\cite{hasegawa1997recursion} and constructions of
free traced monoidal categories were provided
in~\cite{abramsky:traced-compact-closed,hasegawa2008finite}.

\subsection{Nets}
\label{sec:def-nets}
A \emph{signature} $\Sigma=(\Sigma,\sigma,\tau)$ consists of a set~$\Sigma$ of
\emph{symbols} and two functions \hbox{$\sigma,\tau:\Sigma\to\N$}, which to
every symbol~$\alpha$ associate its \emph{arity} and~\emph{coarity} respectively
-- we thus sometimes write~$\alpha:\sigma(\alpha)\to\tau(\alpha)$: the symbols
should be thought as possible building blocks for a process network, with
specified number of inputs and outputs. Given such a signature, a net
consists of instances of symbols (called operators) whose inputs and outputs are
linked together through wires (called ports) defined as follows. Given an
integer~$n$, we write~$\intset{n}$ for the set~$\set{0,\ldots,n-1}$.



\begin{definition}[Net]
  A \emph{net} $N=(P,O,\lambda,s,t)$ from~$m$ to~$n$, with $m,n\in\N$, consists
  of
  \begin{itemize}
  \item a finite set~$P$ of \emph{ports},
  \item a finite set~$O$ of \emph{operators},
  \item a \emph{labeling function} $\lambda:O\to\Sigma$,
  \item a \emph{source function}~$s:S_N\to P$ and an injective \emph{target
      function}~$t:T_N\to P$, where
    \[
    S_N=\setof{(x,i)}{x\in O, i\in\intset{\sigma\circ\lambda(x)}}\uplus\intset{n}
    \qquad\qquad
    T_N=\setof{(x,i)}{x\in O, i\in\intset{\tau\circ\lambda(x)}}\uplus\intset{m}
    \]
  \end{itemize}
  We sometimes write~$N:m\to n$ to indicate that~$N$ is a net from~$m$ to~$n$.
\end{definition}

\begin{example}
  Suppose that~$\Sigma$ contains two symbols~$\alpha$ and~$\beta$ whose sources
  (given by~$\sigma$) are both~$2$ and targets (given by~$\tau$) are
  respectively~$1$ and~$2$. The net drawn in the introduction of this section
  can be formalized as a net~$N:2\to 2$ defined by $P=\set{p_0,\ldots,p_4}$,
  $O=\set{x_0,x_1}$, $\lambda(x_0)=\alpha$, $\lambda(x_1)=\beta$, $s$ and~$t$
  are defined by
\end{example} 
\[
s(x_0,0)=p_0
\qquad
s(x_0,1)=p_4
\qquad
s(x_1,0)=p_2
\qquad
s(x_1,1)=p_1
\qquad
s(0)=p_3
\qquad
s(1)=p_4
\]
\[
\text{and}\qquad
t(x_0,0)=p_2
\qquad
t(x_1,0)=p_3
\qquad
t(x_1,1)=p_4
\qquad
t(0)=p_0
\qquad
t(1)=p_1
\]

\begin{wrapfigure}{r}{4cm}
  \vspace{-4ex}
  \begin{tikzpicture}[baseline=(current bounding box.center),xscale=0.50,yscale=0.50,scale=0.80]
\useasboundingbox (-0.5,-0.5) rectangle (9.5,4.5);
\draw[,,dotted] (0.00,4.00) -- (0.90,4.00);
\draw[] (1.00,4.00) -- (1.05,4.01) -- (1.11,4.01) -- (1.16,4.02) -- (1.21,4.03) -- (1.26,4.03) -- (1.32,4.04) -- (1.37,4.04) -- (1.42,4.05) -- (1.47,4.05) -- (1.52,4.05) -- (1.57,4.06) -- (1.62,4.06) -- (1.67,4.05) -- (1.72,4.05) -- (1.77,4.05) -- (1.82,4.04) -- (1.86,4.03) -- (1.91,4.03) -- (1.96,4.01) -- (2.00,4.00);
\draw[] (2.00,4.00) -- (2.06,3.98) -- (2.12,3.95) -- (2.18,3.92) -- (2.23,3.88) -- (2.29,3.85) -- (2.34,3.80) -- (2.40,3.76) -- (2.45,3.71) -- (2.50,3.66) -- (2.55,3.61) -- (2.59,3.56) -- (2.64,3.50) -- (2.69,3.44) -- (2.73,3.38) -- (2.78,3.32) -- (2.82,3.26) -- (2.87,3.19) -- (2.91,3.13) -- (2.96,3.06) -- (3.00,3.00);
\draw[] (3.00,3.00) -- (3.05,3.01) -- (3.11,3.01) -- (3.16,3.02) -- (3.21,3.03) -- (3.26,3.03) -- (3.32,3.04) -- (3.37,3.04) -- (3.42,3.05) -- (3.47,3.05) -- (3.52,3.05) -- (3.57,3.06) -- (3.62,3.06) -- (3.67,3.05) -- (3.72,3.05) -- (3.77,3.05) -- (3.82,3.04) -- (3.86,3.03) -- (3.91,3.03) -- (3.96,3.01) -- (4.00,3.00);
\draw[] (4.00,3.00) -- (4.06,2.98) -- (4.12,2.95) -- (4.18,2.92) -- (4.23,2.88) -- (4.29,2.85) -- (4.34,2.80) -- (4.40,2.76) -- (4.45,2.71) -- (4.50,2.66) -- (4.55,2.61) -- (4.59,2.56) -- (4.64,2.50) -- (4.69,2.44) -- (4.73,2.38) -- (4.78,2.32) -- (4.82,2.26) -- (4.87,2.19) -- (4.91,2.13) -- (4.96,2.06) -- (5.00,2.00);
\draw[] (5.00,2.00) -- (5.04,1.93) -- (5.09,1.86) -- (5.13,1.80) -- (5.18,1.73) -- (5.22,1.67) -- (5.27,1.60) -- (5.32,1.54) -- (5.36,1.48) -- (5.41,1.42) -- (5.46,1.37) -- (5.51,1.31) -- (5.56,1.26) -- (5.61,1.22) -- (5.66,1.17) -- (5.71,1.13) -- (5.77,1.10) -- (5.82,1.07) -- (5.88,1.04) -- (5.94,1.02) -- (6.00,1.00);
\draw[] (6.00,1.00) -- (6.04,0.99) -- (6.09,0.99) -- (6.13,0.98) -- (6.18,0.98) -- (6.23,0.98) -- (6.27,0.98) -- (6.32,0.98) -- (6.37,0.98) -- (6.42,0.99) -- (6.47,0.99) -- (6.52,1.00) -- (6.57,1.00) -- (6.62,1.00) -- (6.68,1.01) -- (6.73,1.01) -- (6.78,1.01) -- (6.84,1.01) -- (6.89,1.01) -- (6.94,1.01) -- (7.00,1.00);
\draw[] (7.00,1.00) -- (7.06,0.99) -- (7.12,0.98) -- (7.19,0.96) -- (7.25,0.95) -- (7.31,0.93) -- (7.37,0.91) -- (7.43,0.88) -- (7.49,0.86) -- (7.55,0.83) -- (7.60,0.81) -- (7.65,0.78) -- (7.70,0.75) -- (7.75,0.72) -- (7.80,0.69) -- (7.84,0.66) -- (7.88,0.62) -- (7.92,0.59) -- (7.95,0.56) -- (7.98,0.53) -- (8.00,0.50);
\draw[] (8.00,0.50) -- (8.04,0.42) -- (8.05,0.35) -- (8.04,0.28) -- (7.99,0.22) -- (7.92,0.16) -- (7.82,0.11) -- (7.70,0.06) -- (7.57,0.02) -- (7.41,-0.01) -- (7.24,-0.04) -- (7.05,-0.06) -- (6.85,-0.08) -- (6.63,-0.09) -- (6.41,-0.10) -- (6.19,-0.10) -- (5.95,-0.09) -- (5.71,-0.08) -- (5.48,-0.06) -- (5.24,-0.03) -- (5.00,0.00);
\draw[] (5.00,0.00) -- (4.72,0.05) -- (4.45,0.10) -- (4.19,0.16) -- (3.94,0.23) -- (3.70,0.31) -- (3.46,0.39) -- (3.24,0.48) -- (3.04,0.58) -- (2.84,0.68) -- (2.67,0.78) -- (2.51,0.89) -- (2.36,1.01) -- (2.24,1.12) -- (2.14,1.24) -- (2.05,1.37) -- (1.99,1.49) -- (1.96,1.62) -- (1.95,1.74) -- (1.96,1.87) -- (2.00,2.00);
\draw[] (2.00,2.00) -- (2.02,2.05) -- (2.05,2.10) -- (2.08,2.15) -- (2.12,2.20) -- (2.16,2.25) -- (2.20,2.30) -- (2.24,2.35) -- (2.29,2.40) -- (2.34,2.45) -- (2.39,2.50) -- (2.45,2.55) -- (2.50,2.60) -- (2.56,2.65) -- (2.62,2.70) -- (2.68,2.75) -- (2.74,2.80) -- (2.81,2.85) -- (2.87,2.90) -- (2.94,2.95) -- (3.00,3.00);
\draw[] (7.00,3.00) -- (6.95,3.01) -- (6.89,3.01) -- (6.84,3.02) -- (6.79,3.03) -- (6.74,3.03) -- (6.68,3.04) -- (6.63,3.04) -- (6.58,3.05) -- (6.53,3.05) -- (6.48,3.05) -- (6.43,3.06) -- (6.38,3.06) -- (6.33,3.05) -- (6.28,3.05) -- (6.23,3.05) -- (6.18,3.04) -- (6.14,3.03) -- (6.09,3.03) -- (6.04,3.01) -- (6.00,3.00);
\draw[] (6.00,3.00) -- (5.94,2.98) -- (5.88,2.95) -- (5.82,2.92) -- (5.77,2.88) -- (5.71,2.85) -- (5.66,2.80) -- (5.60,2.76) -- (5.55,2.71) -- (5.50,2.66) -- (5.45,2.61) -- (5.41,2.56) -- (5.36,2.50) -- (5.31,2.44) -- (5.27,2.38) -- (5.22,2.32) -- (5.18,2.26) -- (5.13,2.19) -- (5.09,2.13) -- (5.04,2.06) -- (5.00,2.00);
\draw[,,dotted] (9.00,3.00) -- (7.10,3.00);
\draw[] (1.00,1.00) -- (1.16,0.96) -- (1.33,0.93) -- (1.49,0.89) -- (1.66,0.86) -- (1.82,0.83) -- (1.98,0.80) -- (2.14,0.78) -- (2.30,0.76) -- (2.46,0.74) -- (2.61,0.73) -- (2.76,0.72) -- (2.91,0.72) -- (3.06,0.73) -- (3.21,0.74) -- (3.35,0.76) -- (3.49,0.79) -- (3.62,0.83) -- (3.75,0.88) -- (3.88,0.93) -- (4.00,1.00);
\draw[] (4.00,1.00) -- (4.06,1.04) -- (4.11,1.07) -- (4.17,1.11) -- (4.22,1.15) -- (4.27,1.20) -- (4.32,1.24) -- (4.37,1.29) -- (4.43,1.34) -- (4.48,1.39) -- (4.52,1.44) -- (4.57,1.49) -- (4.62,1.55) -- (4.67,1.60) -- (4.72,1.66) -- (4.77,1.71) -- (4.81,1.77) -- (4.86,1.83) -- (4.91,1.88) -- (4.95,1.94) -- (5.00,2.00);
\draw[,,dotted] (0.00,1.00) -- (0.90,1.00);
\draw[,,dotted] (9.00,1.00) -- (7.10,1.00);
\filldraw[fill=white] (2.50, 2.00)  -- (2.50,4.00) -- (3.50,4.00) -- (3.50,2.00) -- (2.50,2.00);
\filldraw[fill=white] (4.50, 1.00)  -- (4.50,3.00) -- (5.50,3.00) -- (5.50,1.00) -- (4.50,1.00);
\draw (0.00,4.00) node{$\bullet$};
\draw (1.00,3.50) node{$p_0$};
\draw (1.00,4.00) node{$\bullet$};
\draw (4.00,3.50) node{$p_2$};
\draw (3.00,3.00) node{$x_0$};
\draw (4.00,3.00) node{$\bullet$};
\draw (7.00,3.00) node{$\bullet$};
\draw (7.00,3.50) node{$p_3$};
\draw (9.00,3.00) node{$\bullet$};
\draw (5.00,2.00) node{$x_1$};
\draw (0.00,1.00) node{$\bullet$};
\draw (1.00,0.50) node{$p_1$};
\draw (1.00,1.00) node{$\bullet$};
\draw (7.00,1.00) node{$\bullet$};
\draw (7.00,1.50) node{$p_4$};
\draw (9.00,1.00) node{$\bullet$};
\end{tikzpicture}
  \vspace{-5ex}
\end{wrapfigure}
\noindent
Graphically, this can be pictured as on the right. The bullets on the left and
on the right indicate the source and target~$m$ and~$n$ and the dotted lines
represent the induced source and target functions~$\intset{m}\to P$
and~$\intset{n}\to P$ respectively. Notice that a port can be used as input for
multiple wires as it is the case for the port~$p_4$ in the example. However, $t$
being injective, two wires cannot have the same output port.

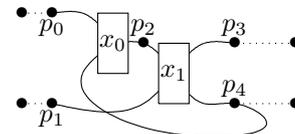
\begin{definition}
  A \emph{morphism}~$\varphi:M\to N$ between two nets~$M,N:m\to n$ (with the
  same source and target) consists of a pair of functions~$\varphi_P:P_M\to P_N$
  and \hbox{$\varphi_O:O_M\to O_N$} such that for every operator~$x\in O_M$,
  $\lambda_N(\varphi_O(x))=\lambda_M(x)$, for every source~$(x,i)\in S_M$,
  \hbox{$\varphi_P(s_M(x,i))=s_N(\varphi_O(x),i)$}, for every~$k\in\intset{n}$,
  $\varphi_P(s_N(k))=s_M(k)$ and similar equations for target functions. Two
  nets~$M$ and~$N$ are isomorphic when there exists an invertible morphism
  between them, which we write~$M\approx N$.
\end{definition}

\begin{definition}
  \label{def:net-constr}
  In order to define a category whose objects are integers and morphisms are
  nets (considered up to isomorphism), we introduce the following constructions:
  \begin{itemize}
  \item \emph{Identity}. The identity net $N:n\to n$ is the net such
    that~$P=\intset{n}$, $O=\emptyset$ and $s,t:\intset{n}\to P$ are both
    the identity function.
  \item \emph{Composition}. Given two nets~$M:n_1\to n_2$ and~$N:n_2\to n_3$,
    their composite \hbox{$N\circ M:n_1\to n_3$} is the net defined by
    $P=P_M\uplus P_N/\!\sim$ where~$\sim$ is the smallest equivalence relation
    such that $s_M(k)\sim t_N(k)$ for every $k\in\intset{n_2}$, $O=O_M\uplus
    O_N$, $\lambda=\lambda_M\uplus\lambda_N$, $s$ is defined by
    $s(x,i)=s_M(x,i)$ if $x\in O_M$, $s(x,i)=s_N(x,i)$ if $x\in O_N$ and
    $s(k)=s_N(k)$ if $k\in\intset{n_3}$, and $t$ is defined similarly.
  \item \emph{Tensor}. Given two nets~$M:m\to m'$ and~$N:n\to n'$, their tensor
    product net \hbox{$M\otimes N:m+n\to m'+n'$} is the net which is defined by
    $P=P_M\uplus P_N$, $O=O_M\uplus O_N$,
    \hbox{$\lambda=\lambda_M\uplus\lambda_N$}, $s$ is defined by
    $s(x,i)=s_M(x,i)$ if $x\in O_M$ and $s(x,i)=s_N(x,i)$ if $x\in O_N$,
    $s(k)=s_M(k)$ if $k\in\intset{m'}$ and $s(k)=s_N(k-m')$
    if~$k\in\intset{n'}$, and $t$ is defined similarly.
  \item \emph{Trace}. Given a net~$N:n_1+n\to n_2+n$, we define the
    net~$\Tr{n_1,n_2}n(N):n_1\to n_2$ by $P=P_N/\sim$ where~$\sim$ is the
    smallest equivalence relation such that $s_N(n_2+k)=t_N(n_1+k)$ for every
    $k\in\intset{n}$, $O=O_N$, $\lambda=\lambda_N$, $s=q\circ s_N$ and $t=q\circ
    t_N$ where~$q:P_N\to P$ is the canonical quotient map.
  \end{itemize}
\end{definition}

It can easily be shown that the constructions above are compatible with
isomorphisms of nets (\eg if $M\approx M'$ and~$N\approx N'$ then~$M\otimes
N\approx M'\otimes N'$). It thus makes sense to define the following category:

\begin{definition}
  We write~$\Net_\Sigma$ for the category~$\Net_\Sigma$ whose objects are
  natural integers, morphisms $N:m\to n$ are isomorphism classes of nets,
  identities and composition are given by the constructions of
  Definition~\ref{def:net-constr}.
\end{definition}

\begin{lemma}
  The category~$\Net_\Sigma$ is well-defined and is monoidal with the tensor
  product of Definition~\ref{def:net-constr} and $0$ as unit.
\end{lemma}

\begin{remark}
  In order to make a more fine-grained study of the categorical structure of
  nets, we could have avoided quotienting morphisms by isomorphisms of net and
  defined a bicategory~\cite{benabou:bicategories} whose 0-cells are integers,
  1-cells are nets and 2-cells are morphisms of nets. We did not do this here to
  simplify the presentation.
\end{remark}

\begin{remark}
  This construction, as well as the following, can be extended without
  difficulty to define multisorted nets (\ie where the various inputs of
  operators have different types), see~\cite{hasegawa2008finite} for a similar
  construction. A nice and abstract description of this construction can be
  carried on using polygraphs~\cite{burroni1993higher}, in a way similar
  to~\cite{mimram:ccp-rs}.
\end{remark}

Even though the output of an operator can be duplicated or erased, the
category~$\Net_\Sigma$ fails to have finite products. This is essentially due to
the fact that duplication and erasure are not natural, \eg the following nets
(of type~$1\to 2$ and~$1\to 0$ respectively) are different:
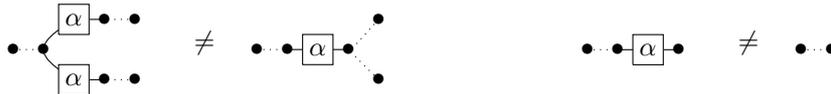
\begin{center}
  \begin{tikzpicture}[baseline=(current bounding box.center),xscale=0.50,yscale=0.50,scale=0.80]
\useasboundingbox (-0.5,-0.5) rectangle (4.5,2.5);
\draw[,,] (2.00,2.00) -- (3.00,2.00);
\draw[] (2.00,0.00) -- (1.93,0.05) -- (1.85,0.10) -- (1.78,0.15) -- (1.70,0.20) -- (1.63,0.25) -- (1.56,0.30) -- (1.50,0.35) -- (1.43,0.40) -- (1.37,0.45) -- (1.31,0.50) -- (1.26,0.55) -- (1.21,0.60) -- (1.16,0.65) -- (1.12,0.70) -- (1.09,0.75) -- (1.06,0.80) -- (1.03,0.85) -- (1.01,0.90) -- (1.00,0.95) -- (1.00,1.00);
\draw[] (1.00,1.00) -- (1.00,1.05) -- (1.01,1.10) -- (1.03,1.15) -- (1.06,1.20) -- (1.09,1.25) -- (1.12,1.30) -- (1.16,1.35) -- (1.21,1.40) -- (1.26,1.45) -- (1.31,1.50) -- (1.37,1.55) -- (1.43,1.60) -- (1.50,1.65) -- (1.56,1.70) -- (1.63,1.75) -- (1.70,1.80) -- (1.78,1.85) -- (1.85,1.90) -- (1.93,1.95) -- (2.00,2.00);
\draw[,,dotted] (4.00,2.00) -- (3.10,2.00);
\draw[,,dotted] (0.00,1.00) -- (0.90,1.00);
\draw[,,] (2.00,0.00) -- (3.00,0.00);
\draw[,,dotted] (4.00,0.00) -- (3.10,0.00);
\filldraw[fill=white] (1.50, 1.50)  -- (1.50,2.50) -- (2.50,2.50) -- (2.50,1.50) -- (1.50,1.50);
\filldraw[fill=white] (1.50, -0.50)  -- (1.50,0.50) -- (2.50,0.50) -- (2.50,-0.50) -- (1.50,-0.50);
\draw (2.00,2.00) node{$\alpha$};
\draw (3.00,2.00) node{$\bullet$};
\draw (4.00,2.00) node{$\bullet$};
\draw (0.00,1.00) node{$\bullet$};
\draw (1.00,1.00) node{$\bullet$};
\draw (2.00,0.00) node{$\alpha$};
\draw (3.00,0.00) node{$\bullet$};
\draw (4.00,0.00) node{$\bullet$};
\end{tikzpicture}
  \quad$\neq$\quad
  \begin{tikzpicture}[baseline=(current bounding box.center),xscale=0.50,yscale=0.50,scale=0.80]
\useasboundingbox (-0.5,-0.5) rectangle (4.5,2.5);
\draw[,,dotted] (4.00,2.00) -- (3.10,1.00);
\draw[,,dotted] (0.00,1.00) -- (0.90,1.00);
\draw[,,] (2.00,1.00) -- (3.00,1.00);
\draw[,,] (1.00,1.00) -- (2.00,1.00);
\draw[,,dotted] (4.00,0.00) -- (3.10,0.90);
\filldraw[fill=white] (1.50, 0.50)  -- (1.50,1.50) -- (2.50,1.50) -- (2.50,0.50) -- (1.50,0.50);
\draw (4.00,2.00) node{$\bullet$};
\draw (0.00,1.00) node{$\bullet$};
\draw (1.00,1.00) node{$\bullet$};
\draw (2.00,1.00) node{$\alpha$};
\draw (3.00,1.00) node{$\bullet$};
\draw (4.00,0.00) node{$\bullet$};
\end{tikzpicture}
  \qquad\qquad\qquad
  \begin{tikzpicture}[baseline=(current bounding box.center),xscale=0.50,yscale=0.50,scale=0.80]
\useasboundingbox (-0.5,-0.5) rectangle (3.5,0.5);
\draw[,,dotted] (0.00,0.00) -- (0.90,0.00);
\draw[,,] (2.00,0.00) -- (3.00,0.00);
\draw[,,] (1.00,0.00) -- (2.00,0.00);
\filldraw[fill=white] (1.50, -0.50)  -- (1.50,0.50) -- (2.50,0.50) -- (2.50,-0.50) -- (1.50,-0.50);
\draw (0.00,0.00) node{$\bullet$};
\draw (1.00,0.00) node{$\bullet$};
\draw (2.00,0.00) node{$\alpha$};
\draw (3.00,0.00) node{$\bullet$};
\end{tikzpicture}
  \quad$\neq$\quad
  \begin{tikzpicture}[baseline=(current bounding box.center),xscale=0.50,yscale=0.50,scale=0.80]
\useasboundingbox (-0.5,-0.5) rectangle (1.5,0.5);
\draw[,,dotted] (0.00,0.00) -- (0.90,0.00);
\draw (0.00,0.00) node{$\bullet$};
\draw (1.00,0.00) node{$\bullet$};
\end{tikzpicture}

\end{center}
We can however recover products by considering the two inequalities above as
rewriting rules on nets from left to right as follows.
\begin{itemize}
\item \emph{Sharing}. Given a net \hbox{$N:m\to n$}, suppose that there exists
  two operators~$x,y\in O$ with the same label and the same inputs, \ie
  $\lambda(x)=\lambda(y)$ and for every $i\in\intset{\sigma\circ\lambda(x)}$,
  $s(x,i)=s(y,i)$. We define a net \hbox{$N':m\to n$} by $P=P_N/\sim$ where
  $\sim$ is the smallest equivalence relation such that \hbox{$t(x,i)\sim
    t(y,i)$} for every $i\in\intset{\tau\circ\lambda(x)}$, $O=O_N/\sim'$
  where~$\sim'$ is the smallest equivalence relation identifying~$x$ and~$y$,
  and $\lambda$, $s$ and~$t$ are obtained by quotienting the corresponding
  functions of~$N$. The net~$N'$ is said to be obtained from~$N$ by
  \emph{sharing}.
\item \emph{Erasing}. Given a net \hbox{$N:m\to n$}, suppose that there exists
  an operator~$x\in O$ such that for every $i\in\intset{\tau\circ\lambda(x)}$,
  $s^{-1}(t(x,i))=\emptyset$. Informally, none of the outputs of the
  operator~$x$ is used as an input for some other operator. We write
  \hbox{$N':m\to n$} for the net obtained from~$N$ by removing the operator~$x$
  as well as all the ports~$t(x,i)$ for $i\in\intset{\tau\circ\lambda(x)}$. The
  net~$N'$ is said to be obtained from~$N$ by \emph{erasing}.
\end{itemize}
We say that $N$ \emph{se-rewrites} to~$N'$ when~$N'$ can be obtained from~$N$ by
sharing or erasing. The \emph{se-equivalence} is the smallest equivalence
relation containing the se-rewriting relation.

\begin{proposition}
  The category~$\sNet_\Sigma$ obtained from~$\Net_\Sigma$ by quotienting
  morphisms by se-equivalence has finite products, given on objects by the
  tensor product of~$\Net_\Sigma$.
\end{proposition}
\begin{proof}
  The terminal object is~$0$ and the product of two objects~$m$ and~$n$ is~$m+n$
  with the projection of~$m$ defined as the net~$N:m+n\to m$ such
  that~$P=\intset{m+n}$, $O=\emptyset$, $s:\intset{m}\to P$ is the canonical
  injection and~$t:\intset{m+n}\to P$ is the identity (and the projection on~$n$
  is defined similarly). All required axioms are easily verified.
\end{proof}

\noindent
It can be shown that the se-rewriting rules form a terminating (they decrease
the number of operators) and confluent rewriting system. The normal forms are
nets which do not contain two operators with the same label and input ports, and
do not contain operators such that none of the outputs are inputs for some other
operator. A direct alternative description of nets modulo se-equivalence, called
shared nets, can thus be defined as follows.

\begin{definition}
  A \emph{shared net} $N=(P,O,s,t)$ from~$m$ to~$n$ consists of
  \begin{itemize}
  \item a finite set~$P$ of \emph{ports},
  \item a finite set~$O$ of \emph{operators} which are pairs
    $(\alpha,(s_i)_{i\in\intset{\sigma(\alpha)}})$ where~$\alpha\in\Sigma$ is a
    symbol and $(s_i)_{i\in\intset{\sigma(\alpha)}}$ is a family of ports called
    the \emph{sources} of the operators,
  \item a \emph{source function}~$s:\intset{n}\to P$,
  \item an injective \emph{target function}~$t:T_N\to P$, where $T_N=\setof{(x,i)}{x\in
      O,i\in\intset{\tau\circ\lambda(x)}}\uplus\intset{m}$,
  \end{itemize}
  such that for every operator~$x\in O$, $s^{-1}(t(T_x))\neq\emptyset$ where
  $T_x=\setof{(x,i)}{i\in\intset{\tau\circ\lambda(x)}}$.
\end{definition}

\begin{proposition}
  A category whose objects are integers and morphisms are shared nets modulo
  (suitably defined) isomorphism can be defined in a similar way as previously,
  and this category can be shown to be equivalent to~$\sNet_\Sigma$ through
  product-preserving functors.
\end{proposition}
\begin{proof}
  The canonical forms of nets \wrt se-rewriting are in bijection with shared
  nets.
\end{proof}

\noindent
Next section justifies why the category $\sNet$ provides a convincing definition
of the networks.

\subsection{Nets as free fixpoint categories}
\label{sec:fixpoint-cat}
We now study the structure of the category~$\sNet_\Sigma$ in order to define a
proper denotational model this category. Recall that a (strict) \emph{monoidal
  category} $(\C,\otimes,I)$ consists of a category~$\C$ together with a
bifunctor $\otimes:\C\times\C\to\C$, called \emph{tensor}, and an object~$I$,
called \emph{unit}, such that the tensor is strictly associative and admits
units as neutral elements. A (strict) \emph{symmetric monoidal category}
consists of a monoidal category $(\C,\otimes,I)$ equipped with a natural family
\hbox{$\gamma_{A,B}:A\otimes B\to B\otimes A$} of isomorphisms satisfying
$\gamma_{B,A}\circ\gamma_{A,B}=\id_{A\otimes B}$ as well as other coherence
axioms, see~\cite{maclane:cwm} for details. Any category~$\C$ with finite
products admits a structure of symmetric monoidal category with the cartesian
product~$\times$ as tensor and the terminal object~$1$ as unit, and this
structure can be chosen to be strict in the case of~$\sNet_\Sigma$ (thus we only
consider strict monoidal categories in the following for simplicity). A natural
notion of ``feedback'' was formalized in monoidal categories by Joyal, Street
and Verity~\cite{joyal-street-verity:traced-monoidal-categories} as follows:

\begin{definition}[Trace]
  \label{def:trace}
  A \emph{trace} on a symmetric monoidal category~$\C$ consists of a
  \emph{natural}
\end{definition}

\begin{wrapfigure}{r}{2.5cm}
  \vspace{-4.5ex}
  \begin{tikzpicture}[baseline=(current bounding box.center),xscale=0.50,yscale=0.50,scale=0.80]
\useasboundingbox (-0.5,-0.5) rectangle (6.5,3.5);
\draw[,,] (1.00,2.00) -- (3.00,2.00);
\draw[] (3.00,1.00) -- (3.05,1.01) -- (3.09,1.01) -- (3.14,1.02) -- (3.19,1.02) -- (3.24,1.03) -- (3.28,1.03) -- (3.33,1.04) -- (3.38,1.04) -- (3.43,1.04) -- (3.48,1.04) -- (3.53,1.04) -- (3.58,1.04) -- (3.63,1.04) -- (3.68,1.04) -- (3.73,1.04) -- (3.78,1.03) -- (3.84,1.03) -- (3.89,1.02) -- (3.94,1.01) -- (4.00,1.00);
\draw[] (4.00,1.00) -- (4.06,0.99) -- (4.13,0.97) -- (4.19,0.95) -- (4.25,0.93) -- (4.32,0.91) -- (4.38,0.89) -- (4.44,0.86) -- (4.50,0.84) -- (4.56,0.81) -- (4.62,0.79) -- (4.67,0.76) -- (4.72,0.73) -- (4.77,0.70) -- (4.82,0.67) -- (4.86,0.64) -- (4.90,0.61) -- (4.93,0.58) -- (4.96,0.56) -- (4.98,0.53) -- (5.00,0.50);
\draw[] (5.00,0.50) -- (5.02,0.45) -- (5.02,0.40) -- (5.00,0.36) -- (4.97,0.32) -- (4.92,0.28) -- (4.86,0.24) -- (4.78,0.21) -- (4.69,0.18) -- (4.59,0.15) -- (4.48,0.12) -- (4.36,0.10) -- (4.23,0.08) -- (4.10,0.06) -- (3.95,0.04) -- (3.80,0.03) -- (3.65,0.02) -- (3.49,0.01) -- (3.33,0.00) -- (3.16,0.00) -- (3.00,0.00);
\draw[] (3.00,0.00) -- (2.84,0.00) -- (2.67,0.00) -- (2.51,0.01) -- (2.35,0.02) -- (2.20,0.03) -- (2.05,0.04) -- (1.90,0.06) -- (1.77,0.08) -- (1.64,0.10) -- (1.52,0.12) -- (1.41,0.15) -- (1.31,0.18) -- (1.22,0.21) -- (1.14,0.24) -- (1.08,0.28) -- (1.03,0.32) -- (1.00,0.36) -- (0.98,0.40) -- (0.98,0.45) -- (1.00,0.50);
\draw[] (1.00,0.50) -- (1.02,0.53) -- (1.04,0.56) -- (1.07,0.58) -- (1.10,0.61) -- (1.14,0.64) -- (1.18,0.67) -- (1.23,0.70) -- (1.28,0.73) -- (1.33,0.76) -- (1.38,0.79) -- (1.44,0.81) -- (1.50,0.84) -- (1.56,0.86) -- (1.62,0.89) -- (1.68,0.91) -- (1.75,0.93) -- (1.81,0.95) -- (1.87,0.97) -- (1.94,0.99) -- (2.00,1.00);
\draw[] (2.00,1.00) -- (2.06,1.01) -- (2.11,1.02) -- (2.16,1.03) -- (2.22,1.03) -- (2.27,1.04) -- (2.32,1.04) -- (2.37,1.04) -- (2.42,1.04) -- (2.47,1.04) -- (2.52,1.04) -- (2.57,1.04) -- (2.62,1.04) -- (2.67,1.04) -- (2.72,1.03) -- (2.76,1.03) -- (2.81,1.02) -- (2.86,1.02) -- (2.91,1.01) -- (2.95,1.01) -- (3.00,1.00);
\draw[,,] (5.00,2.00) -- (3.00,2.00);
\draw[dashed,] (2.00,2.75) rectangle (4.00,0.30);
\filldraw[,fill=white] (2.50,0.50) rectangle (3.50,2.50);
\draw (0.50,2.00) node{$A$};
\draw (3.00,2.00) node{$f$};
\draw (5.50,2.00) node{$B$};
\draw (0.90,1.00) node{$X$};
\end{tikzpicture}
  \vspace{-5ex}
\end{wrapfigure}
\vspace{-1ex}
\noindent
family of functions~$\Tr{A,B}X:\C(A\otimes X,B\otimes X)\to\C(A,B)$. Given a
morphism \hbox{$f:A\otimes X\to B\otimes X$}, the morphism~$\Tr{A,B}X(f):A\to B$
is often drawn as on the right. A trace should satisfy the following axioms.
  \begin{enumerate}
  \item \emph{Vanishing}: for every~$f:A\otimes X\otimes Y\to B\otimes X\otimes
    Y$, $\Tr{A,B}{X\otimes Y}(f)=\Tr{A,B}X(\Tr{A\otimes X,B\otimes X}Y(f))$
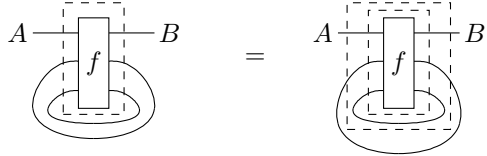
    \begin{center}
      \begin{tikzpicture}[baseline=(current bounding box.center),xscale=0.50,yscale=0.50,scale=0.80]
\useasboundingbox (-0.5,-0.5) rectangle (6.5,4.5);
\draw[] (3.00,2.00) -- (2.95,2.01) -- (2.90,2.02) -- (2.85,2.02) -- (2.80,2.03) -- (2.74,2.04) -- (2.69,2.05) -- (2.64,2.05) -- (2.59,2.06) -- (2.54,2.06) -- (2.49,2.06) -- (2.44,2.07) -- (2.39,2.07) -- (2.34,2.06) -- (2.29,2.06) -- (2.24,2.06) -- (2.19,2.05) -- (2.14,2.04) -- (2.10,2.03) -- (2.05,2.02) -- (2.00,2.00);
\draw[] (2.00,2.00) -- (1.91,1.97) -- (1.83,1.92) -- (1.75,1.87) -- (1.67,1.82) -- (1.59,1.76) -- (1.52,1.69) -- (1.45,1.62) -- (1.38,1.54) -- (1.32,1.46) -- (1.26,1.38) -- (1.21,1.29) -- (1.16,1.20) -- (1.11,1.11) -- (1.08,1.02) -- (1.05,0.93) -- (1.02,0.84) -- (1.01,0.75) -- (1.00,0.67) -- (0.99,0.58) -- (1.00,0.50);
\draw[] (1.00,0.50) -- (1.02,0.40) -- (1.05,0.31) -- (1.09,0.22) -- (1.15,0.14) -- (1.21,0.06) -- (1.29,-0.01) -- (1.37,-0.08) -- (1.47,-0.14) -- (1.57,-0.20) -- (1.67,-0.25) -- (1.79,-0.30) -- (1.91,-0.34) -- (2.03,-0.38) -- (2.16,-0.41) -- (2.30,-0.44) -- (2.43,-0.46) -- (2.57,-0.48) -- (2.71,-0.49) -- (2.86,-0.50) -- (3.00,-0.50);
\draw[] (3.00,-0.50) -- (3.14,-0.50) -- (3.29,-0.49) -- (3.43,-0.48) -- (3.57,-0.46) -- (3.70,-0.44) -- (3.84,-0.41) -- (3.97,-0.38) -- (4.09,-0.34) -- (4.21,-0.30) -- (4.33,-0.25) -- (4.43,-0.20) -- (4.53,-0.14) -- (4.63,-0.08) -- (4.71,-0.01) -- (4.79,0.06) -- (4.85,0.14) -- (4.91,0.22) -- (4.95,0.31) -- (4.98,0.40) -- (5.00,0.50);
\draw[] (5.00,0.50) -- (5.01,0.58) -- (5.00,0.67) -- (4.99,0.75) -- (4.98,0.84) -- (4.95,0.93) -- (4.92,1.02) -- (4.89,1.11) -- (4.84,1.20) -- (4.79,1.29) -- (4.74,1.38) -- (4.68,1.46) -- (4.62,1.54) -- (4.55,1.62) -- (4.48,1.69) -- (4.41,1.76) -- (4.33,1.82) -- (4.25,1.87) -- (4.17,1.92) -- (4.09,1.97) -- (4.00,2.00);
\draw[] (4.00,2.00) -- (3.95,2.02) -- (3.90,2.03) -- (3.86,2.04) -- (3.81,2.05) -- (3.76,2.06) -- (3.71,2.06) -- (3.66,2.06) -- (3.61,2.07) -- (3.56,2.07) -- (3.51,2.06) -- (3.46,2.06) -- (3.41,2.06) -- (3.36,2.05) -- (3.31,2.05) -- (3.26,2.04) -- (3.20,2.03) -- (3.15,2.02) -- (3.10,2.02) -- (3.05,2.01) -- (3.00,2.00);
\draw[,,] (1.00,3.00) -- (3.00,3.00);
\draw[] (3.00,1.00) -- (3.05,1.01) -- (3.10,1.02) -- (3.15,1.03) -- (3.20,1.04) -- (3.25,1.05) -- (3.30,1.06) -- (3.35,1.07) -- (3.40,1.08) -- (3.45,1.08) -- (3.50,1.09) -- (3.55,1.09) -- (3.60,1.09) -- (3.65,1.09) -- (3.70,1.08) -- (3.75,1.08) -- (3.80,1.07) -- (3.85,1.05) -- (3.90,1.04) -- (3.95,1.02) -- (4.00,1.00);
\draw[] (4.00,1.00) -- (4.03,0.98) -- (4.07,0.96) -- (4.10,0.94) -- (4.14,0.92) -- (4.17,0.90) -- (4.20,0.88) -- (4.23,0.85) -- (4.26,0.83) -- (4.29,0.80) -- (4.32,0.77) -- (4.35,0.75) -- (4.37,0.72) -- (4.39,0.69) -- (4.42,0.66) -- (4.44,0.63) -- (4.45,0.61) -- (4.47,0.58) -- (4.48,0.55) -- (4.49,0.53) -- (4.50,0.50);
\draw[] (4.50,0.50) -- (4.51,0.44) -- (4.50,0.39) -- (4.49,0.35) -- (4.46,0.30) -- (4.42,0.26) -- (4.37,0.22) -- (4.31,0.19) -- (4.25,0.16) -- (4.17,0.13) -- (4.09,0.11) -- (4.00,0.09) -- (3.90,0.07) -- (3.80,0.05) -- (3.70,0.04) -- (3.59,0.02) -- (3.47,0.02) -- (3.36,0.01) -- (3.24,0.00) -- (3.12,0.00) -- (3.00,0.00);
\draw[] (3.00,0.00) -- (2.88,0.00) -- (2.76,0.00) -- (2.64,0.01) -- (2.53,0.02) -- (2.41,0.02) -- (2.30,0.04) -- (2.20,0.05) -- (2.10,0.07) -- (2.00,0.09) -- (1.91,0.11) -- (1.83,0.13) -- (1.75,0.16) -- (1.69,0.19) -- (1.63,0.22) -- (1.58,0.26) -- (1.54,0.30) -- (1.51,0.35) -- (1.50,0.39) -- (1.49,0.44) -- (1.50,0.50);
\draw[] (1.50,0.50) -- (1.51,0.53) -- (1.52,0.55) -- (1.53,0.58) -- (1.55,0.61) -- (1.56,0.63) -- (1.58,0.66) -- (1.61,0.69) -- (1.63,0.72) -- (1.65,0.75) -- (1.68,0.77) -- (1.71,0.80) -- (1.74,0.83) -- (1.77,0.85) -- (1.80,0.88) -- (1.83,0.90) -- (1.86,0.92) -- (1.90,0.94) -- (1.93,0.96) -- (1.97,0.98) -- (2.00,1.00);
\draw[] (2.00,1.00) -- (2.05,1.02) -- (2.10,1.04) -- (2.15,1.05) -- (2.20,1.07) -- (2.25,1.08) -- (2.30,1.08) -- (2.35,1.09) -- (2.40,1.09) -- (2.45,1.09) -- (2.50,1.09) -- (2.55,1.08) -- (2.60,1.08) -- (2.65,1.07) -- (2.70,1.06) -- (2.75,1.05) -- (2.80,1.04) -- (2.85,1.03) -- (2.90,1.02) -- (2.95,1.01) -- (3.00,1.00);
\draw[,,] (5.00,3.00) -- (3.00,3.00);
\draw[dashed,] (2.00,4.00) rectangle (4.00,0.30);
\filldraw[,fill=white] (2.50,0.50) rectangle (3.50,3.50);
\draw (0.50,3.00) node{$A$};
\draw (5.50,3.00) node{$B$};
\draw (3.00,2.00) node{$f$};
\end{tikzpicture}
      \qeq
      \begin{tikzpicture}[baseline=(current bounding box.center),xscale=0.50,yscale=0.50,scale=0.80]
\useasboundingbox (-0.5,-0.5) rectangle (6.5,4.5);
\draw[] (3.00,2.00) -- (2.95,2.01) -- (2.90,2.02) -- (2.84,2.02) -- (2.79,2.03) -- (2.74,2.04) -- (2.69,2.04) -- (2.64,2.05) -- (2.59,2.06) -- (2.54,2.06) -- (2.49,2.06) -- (2.44,2.06) -- (2.39,2.06) -- (2.34,2.06) -- (2.29,2.06) -- (2.24,2.05) -- (2.19,2.05) -- (2.14,2.04) -- (2.09,2.03) -- (2.05,2.02) -- (2.00,2.00);
\draw[] (2.00,2.00) -- (1.90,1.96) -- (1.80,1.90) -- (1.70,1.84) -- (1.61,1.76) -- (1.52,1.68) -- (1.44,1.59) -- (1.36,1.49) -- (1.29,1.39) -- (1.23,1.28) -- (1.17,1.17) -- (1.11,1.05) -- (1.07,0.93) -- (1.03,0.81) -- (1.00,0.69) -- (0.98,0.57) -- (0.96,0.45) -- (0.96,0.33) -- (0.96,0.22) -- (0.98,0.11) -- (1.00,0.00);
\draw[] (1.00,0.00) -- (1.03,-0.10) -- (1.08,-0.20) -- (1.13,-0.28) -- (1.20,-0.37) -- (1.27,-0.45) -- (1.35,-0.52) -- (1.44,-0.59) -- (1.53,-0.65) -- (1.63,-0.71) -- (1.74,-0.76) -- (1.85,-0.80) -- (1.96,-0.85) -- (2.08,-0.88) -- (2.21,-0.91) -- (2.33,-0.94) -- (2.46,-0.96) -- (2.60,-0.98) -- (2.73,-0.99) -- (2.86,-1.00) -- (3.00,-1.00);
\draw[] (3.00,-1.00) -- (3.14,-1.00) -- (3.27,-0.99) -- (3.40,-0.98) -- (3.54,-0.96) -- (3.67,-0.94) -- (3.79,-0.91) -- (3.92,-0.88) -- (4.04,-0.85) -- (4.15,-0.80) -- (4.26,-0.76) -- (4.37,-0.71) -- (4.47,-0.65) -- (4.56,-0.59) -- (4.65,-0.52) -- (4.73,-0.45) -- (4.80,-0.37) -- (4.87,-0.28) -- (4.92,-0.20) -- (4.97,-0.10) -- (5.00,0.00);
\draw[] (5.00,0.00) -- (5.02,0.11) -- (5.04,0.22) -- (5.04,0.33) -- (5.04,0.45) -- (5.02,0.57) -- (5.00,0.69) -- (4.97,0.81) -- (4.93,0.93) -- (4.89,1.05) -- (4.83,1.17) -- (4.77,1.28) -- (4.71,1.39) -- (4.64,1.49) -- (4.56,1.59) -- (4.48,1.68) -- (4.39,1.76) -- (4.30,1.84) -- (4.20,1.90) -- (4.10,1.96) -- (4.00,2.00);
\draw[] (4.00,2.00) -- (3.95,2.02) -- (3.91,2.03) -- (3.86,2.04) -- (3.81,2.05) -- (3.76,2.05) -- (3.71,2.06) -- (3.66,2.06) -- (3.61,2.06) -- (3.56,2.06) -- (3.51,2.06) -- (3.46,2.06) -- (3.41,2.06) -- (3.36,2.05) -- (3.31,2.04) -- (3.26,2.04) -- (3.21,2.03) -- (3.16,2.02) -- (3.10,2.02) -- (3.05,2.01) -- (3.00,2.00);
\draw[,,] (1.00,3.00) -- (3.00,3.00);
\draw[] (3.00,1.00) -- (3.05,1.01) -- (3.10,1.02) -- (3.15,1.03) -- (3.20,1.04) -- (3.25,1.05) -- (3.30,1.06) -- (3.35,1.07) -- (3.40,1.08) -- (3.45,1.08) -- (3.50,1.09) -- (3.55,1.09) -- (3.60,1.09) -- (3.65,1.09) -- (3.70,1.08) -- (3.75,1.08) -- (3.80,1.07) -- (3.85,1.05) -- (3.90,1.04) -- (3.95,1.02) -- (4.00,1.00);
\draw[] (4.00,1.00) -- (4.03,0.98) -- (4.07,0.96) -- (4.10,0.94) -- (4.14,0.92) -- (4.17,0.90) -- (4.20,0.88) -- (4.23,0.85) -- (4.26,0.83) -- (4.29,0.80) -- (4.32,0.77) -- (4.35,0.75) -- (4.37,0.72) -- (4.39,0.69) -- (4.42,0.66) -- (4.44,0.63) -- (4.45,0.61) -- (4.47,0.58) -- (4.48,0.55) -- (4.49,0.53) -- (4.50,0.50);
\draw[] (4.50,0.50) -- (4.51,0.44) -- (4.50,0.39) -- (4.49,0.35) -- (4.46,0.30) -- (4.42,0.26) -- (4.37,0.22) -- (4.31,0.19) -- (4.25,0.16) -- (4.17,0.13) -- (4.09,0.11) -- (4.00,0.09) -- (3.90,0.07) -- (3.80,0.05) -- (3.70,0.04) -- (3.59,0.02) -- (3.47,0.02) -- (3.36,0.01) -- (3.24,0.00) -- (3.12,0.00) -- (3.00,0.00);
\draw[] (3.00,0.00) -- (2.88,0.00) -- (2.76,0.00) -- (2.64,0.01) -- (2.53,0.02) -- (2.41,0.02) -- (2.30,0.04) -- (2.20,0.05) -- (2.10,0.07) -- (2.00,0.09) -- (1.91,0.11) -- (1.83,0.13) -- (1.75,0.16) -- (1.69,0.19) -- (1.63,0.22) -- (1.58,0.26) -- (1.54,0.30) -- (1.51,0.35) -- (1.50,0.39) -- (1.49,0.44) -- (1.50,0.50);
\draw[] (1.50,0.50) -- (1.51,0.53) -- (1.52,0.55) -- (1.53,0.58) -- (1.55,0.61) -- (1.56,0.63) -- (1.58,0.66) -- (1.61,0.69) -- (1.63,0.72) -- (1.65,0.75) -- (1.68,0.77) -- (1.71,0.80) -- (1.74,0.83) -- (1.77,0.85) -- (1.80,0.88) -- (1.83,0.90) -- (1.86,0.92) -- (1.90,0.94) -- (1.93,0.96) -- (1.97,0.98) -- (2.00,1.00);
\draw[] (2.00,1.00) -- (2.05,1.02) -- (2.10,1.04) -- (2.15,1.05) -- (2.20,1.07) -- (2.25,1.08) -- (2.30,1.08) -- (2.35,1.09) -- (2.40,1.09) -- (2.45,1.09) -- (2.50,1.09) -- (2.55,1.08) -- (2.60,1.08) -- (2.65,1.07) -- (2.70,1.06) -- (2.75,1.05) -- (2.80,1.04) -- (2.85,1.03) -- (2.90,1.02) -- (2.95,1.01) -- (3.00,1.00);
\draw[,,] (5.00,3.00) -- (3.00,3.00);
\draw[dashed,] (1.30,4.00) rectangle (4.70,-0.20);
\draw[dashed,] (2.00,3.75) rectangle (4.00,0.30);
\filldraw[,fill=white] (2.50,0.50) rectangle (3.50,3.50);
\draw (0.50,3.00) node{$A$};
\draw (5.50,3.00) node{$B$};
\draw (3.00,2.00) node{$f$};
\end{tikzpicture}
    \end{center}
  \item \emph{Superposing}:\\
    for every~$f:A\otimes X\to B\otimes X$ and~$g:C\to D$,
    $g\otimes\Tr{A,B}X(f)=\Tr{C\otimes A,D\otimes B}X(g\otimes f)$
    \begin{center}
      \begin{tikzpicture}[baseline=(current bounding box.center),xscale=0.50,yscale=0.50,scale=0.80]
\useasboundingbox (-0.5,-0.5) rectangle (6.5,5.5);
\draw[,,] (1.00,4.00) -- (3.00,4.00);
\draw[,,] (5.00,4.00) -- (3.00,4.00);
\draw[,,] (1.00,2.00) -- (3.00,2.00);
\draw[] (3.00,1.00) -- (3.05,1.01) -- (3.09,1.01) -- (3.14,1.02) -- (3.19,1.02) -- (3.24,1.03) -- (3.28,1.03) -- (3.33,1.04) -- (3.38,1.04) -- (3.43,1.04) -- (3.48,1.04) -- (3.53,1.04) -- (3.58,1.04) -- (3.63,1.04) -- (3.68,1.04) -- (3.73,1.04) -- (3.78,1.03) -- (3.84,1.03) -- (3.89,1.02) -- (3.94,1.01) -- (4.00,1.00);
\draw[] (4.00,1.00) -- (4.06,0.99) -- (4.13,0.97) -- (4.19,0.95) -- (4.25,0.93) -- (4.32,0.91) -- (4.38,0.89) -- (4.44,0.86) -- (4.50,0.84) -- (4.56,0.81) -- (4.62,0.79) -- (4.67,0.76) -- (4.72,0.73) -- (4.77,0.70) -- (4.82,0.67) -- (4.86,0.64) -- (4.90,0.61) -- (4.93,0.58) -- (4.96,0.56) -- (4.98,0.53) -- (5.00,0.50);
\draw[] (5.00,0.50) -- (5.02,0.45) -- (5.02,0.40) -- (5.00,0.36) -- (4.97,0.32) -- (4.92,0.28) -- (4.86,0.24) -- (4.78,0.21) -- (4.69,0.18) -- (4.59,0.15) -- (4.48,0.12) -- (4.36,0.10) -- (4.23,0.08) -- (4.10,0.06) -- (3.95,0.04) -- (3.80,0.03) -- (3.65,0.02) -- (3.49,0.01) -- (3.33,0.00) -- (3.16,0.00) -- (3.00,0.00);
\draw[] (3.00,0.00) -- (2.84,0.00) -- (2.67,0.00) -- (2.51,0.01) -- (2.35,0.02) -- (2.20,0.03) -- (2.05,0.04) -- (1.90,0.06) -- (1.77,0.08) -- (1.64,0.10) -- (1.52,0.12) -- (1.41,0.15) -- (1.31,0.18) -- (1.22,0.21) -- (1.14,0.24) -- (1.08,0.28) -- (1.03,0.32) -- (1.00,0.36) -- (0.98,0.40) -- (0.98,0.45) -- (1.00,0.50);
\draw[] (1.00,0.50) -- (1.02,0.53) -- (1.04,0.56) -- (1.07,0.58) -- (1.10,0.61) -- (1.14,0.64) -- (1.18,0.67) -- (1.23,0.70) -- (1.28,0.73) -- (1.33,0.76) -- (1.38,0.79) -- (1.44,0.81) -- (1.50,0.84) -- (1.56,0.86) -- (1.62,0.89) -- (1.68,0.91) -- (1.75,0.93) -- (1.81,0.95) -- (1.87,0.97) -- (1.94,0.99) -- (2.00,1.00);
\draw[] (2.00,1.00) -- (2.06,1.01) -- (2.11,1.02) -- (2.16,1.03) -- (2.22,1.03) -- (2.27,1.04) -- (2.32,1.04) -- (2.37,1.04) -- (2.42,1.04) -- (2.47,1.04) -- (2.52,1.04) -- (2.57,1.04) -- (2.62,1.04) -- (2.67,1.04) -- (2.72,1.03) -- (2.76,1.03) -- (2.81,1.02) -- (2.86,1.02) -- (2.91,1.01) -- (2.95,1.01) -- (3.00,1.00);
\draw[,,] (5.00,2.00) -- (3.00,2.00);
\filldraw[,fill=white] (2.50,3.50) rectangle (3.50,4.50);
\draw[dashed,] (2.00,3.00) rectangle (4.00,0.30);
\filldraw[,fill=white] (2.50,0.50) rectangle (3.50,2.50);
\draw (0.50,4.00) node{$C$};
\draw (3.00,4.00) node{$g$};
\draw (5.50,4.00) node{$D$};
\draw (0.50,2.00) node{$A$};
\draw (3.00,2.00) node{$f$};
\draw (5.50,2.00) node{$B$};
\end{tikzpicture}
      \qeq
      \begin{tikzpicture}[baseline=(current bounding box.center),xscale=0.50,yscale=0.50,scale=0.80]
\useasboundingbox (-0.5,-0.5) rectangle (6.5,5.5);
\draw[,,] (1.00,4.00) -- (3.00,4.00);
\draw[,,] (5.00,4.00) -- (3.00,4.00);
\draw[,,] (1.00,2.00) -- (3.00,2.00);
\draw[] (3.00,1.00) -- (3.05,1.01) -- (3.09,1.01) -- (3.14,1.02) -- (3.19,1.02) -- (3.24,1.03) -- (3.28,1.03) -- (3.33,1.04) -- (3.38,1.04) -- (3.43,1.04) -- (3.48,1.04) -- (3.53,1.04) -- (3.58,1.04) -- (3.63,1.04) -- (3.68,1.04) -- (3.73,1.04) -- (3.78,1.03) -- (3.84,1.03) -- (3.89,1.02) -- (3.94,1.01) -- (4.00,1.00);
\draw[] (4.00,1.00) -- (4.06,0.99) -- (4.13,0.97) -- (4.19,0.95) -- (4.25,0.93) -- (4.32,0.91) -- (4.38,0.89) -- (4.44,0.86) -- (4.50,0.84) -- (4.56,0.81) -- (4.62,0.79) -- (4.67,0.76) -- (4.72,0.73) -- (4.77,0.70) -- (4.82,0.67) -- (4.86,0.64) -- (4.90,0.61) -- (4.93,0.58) -- (4.96,0.56) -- (4.98,0.53) -- (5.00,0.50);
\draw[] (5.00,0.50) -- (5.02,0.45) -- (5.02,0.40) -- (5.00,0.36) -- (4.97,0.32) -- (4.92,0.28) -- (4.86,0.24) -- (4.78,0.21) -- (4.69,0.18) -- (4.59,0.15) -- (4.48,0.12) -- (4.36,0.10) -- (4.23,0.08) -- (4.10,0.06) -- (3.95,0.04) -- (3.80,0.03) -- (3.65,0.02) -- (3.49,0.01) -- (3.33,0.00) -- (3.16,0.00) -- (3.00,0.00);
\draw[] (3.00,0.00) -- (2.84,0.00) -- (2.67,0.00) -- (2.51,0.01) -- (2.35,0.02) -- (2.20,0.03) -- (2.05,0.04) -- (1.90,0.06) -- (1.77,0.08) -- (1.64,0.10) -- (1.52,0.12) -- (1.41,0.15) -- (1.31,0.18) -- (1.22,0.21) -- (1.14,0.24) -- (1.08,0.28) -- (1.03,0.32) -- (1.00,0.36) -- (0.98,0.40) -- (0.98,0.45) -- (1.00,0.50);
\draw[] (1.00,0.50) -- (1.02,0.53) -- (1.04,0.56) -- (1.07,0.58) -- (1.10,0.61) -- (1.14,0.64) -- (1.18,0.67) -- (1.23,0.70) -- (1.28,0.73) -- (1.33,0.76) -- (1.38,0.79) -- (1.44,0.81) -- (1.50,0.84) -- (1.56,0.86) -- (1.62,0.89) -- (1.68,0.91) -- (1.75,0.93) -- (1.81,0.95) -- (1.87,0.97) -- (1.94,0.99) -- (2.00,1.00);
\draw[] (2.00,1.00) -- (2.06,1.01) -- (2.11,1.02) -- (2.16,1.03) -- (2.22,1.03) -- (2.27,1.04) -- (2.32,1.04) -- (2.37,1.04) -- (2.42,1.04) -- (2.47,1.04) -- (2.52,1.04) -- (2.57,1.04) -- (2.62,1.04) -- (2.67,1.04) -- (2.72,1.03) -- (2.76,1.03) -- (2.81,1.02) -- (2.86,1.02) -- (2.91,1.01) -- (2.95,1.01) -- (3.00,1.00);
\draw[,,] (5.00,2.00) -- (3.00,2.00);
\draw[dashed,] (2.00,4.75) rectangle (4.00,0.30);
\filldraw[,fill=white] (2.50,3.50) rectangle (3.50,4.50);
\filldraw[,fill=white] (2.50,0.50) rectangle (3.50,2.50);
\draw (0.50,4.00) node{$C$};
\draw (3.00,4.00) node{$g$};
\draw (5.50,4.00) node{$D$};
\draw (0.50,2.00) node{$A$};
\draw (3.00,2.00) node{$f$};
\draw (5.50,2.00) node{$B$};
\end{tikzpicture}
    \end{center}
  \item \emph{Yanking}: for every object~$X$, $\Tr{X,X}X(\gamma_{X,X})=\id_X$
    \begin{center}
      \begin{tikzpicture}[baseline=(current bounding box.center),xscale=0.50,yscale=0.50,scale=0.80]
\useasboundingbox (-0.5,-0.5) rectangle (6.5,4.5);
\draw[] (5.00,3.00) -- (4.95,3.01) -- (4.89,3.02) -- (4.84,3.02) -- (4.79,3.03) -- (4.74,3.04) -- (4.68,3.04) -- (4.63,3.05) -- (4.58,3.05) -- (4.53,3.06) -- (4.48,3.06) -- (4.43,3.06) -- (4.38,3.06) -- (4.33,3.06) -- (4.28,3.06) -- (4.23,3.05) -- (4.18,3.04) -- (4.14,3.04) -- (4.09,3.03) -- (4.04,3.01) -- (4.00,3.00);
\draw[] (4.00,3.00) -- (3.94,2.98) -- (3.88,2.95) -- (3.82,2.91) -- (3.76,2.88) -- (3.71,2.84) -- (3.65,2.80) -- (3.60,2.75) -- (3.55,2.70) -- (3.50,2.65) -- (3.45,2.59) -- (3.40,2.54) -- (3.35,2.48) -- (3.31,2.42) -- (3.26,2.36) -- (3.21,2.30) -- (3.17,2.24) -- (3.13,2.18) -- (3.08,2.12) -- (3.04,2.06) -- (3.00,2.00);
\draw[] (3.00,2.00) -- (2.96,1.94) -- (2.92,1.88) -- (2.88,1.82) -- (2.83,1.77) -- (2.79,1.71) -- (2.75,1.66) -- (2.71,1.60) -- (2.66,1.55) -- (2.62,1.50) -- (2.57,1.45) -- (2.52,1.40) -- (2.48,1.35) -- (2.42,1.30) -- (2.37,1.25) -- (2.32,1.21) -- (2.26,1.17) -- (2.20,1.12) -- (2.13,1.08) -- (2.07,1.04) -- (2.00,1.00);
\draw[] (2.00,1.00) -- (1.94,0.97) -- (1.88,0.94) -- (1.83,0.91) -- (1.76,0.88) -- (1.70,0.86) -- (1.64,0.83) -- (1.58,0.80) -- (1.52,0.78) -- (1.46,0.75) -- (1.40,0.73) -- (1.35,0.70) -- (1.30,0.68) -- (1.25,0.66) -- (1.20,0.63) -- (1.15,0.61) -- (1.11,0.59) -- (1.08,0.57) -- (1.05,0.54) -- (1.02,0.52) -- (1.00,0.50);
\draw[] (1.00,0.50) -- (0.98,0.46) -- (0.97,0.42) -- (0.99,0.38) -- (1.02,0.35) -- (1.06,0.31) -- (1.13,0.28) -- (1.20,0.24) -- (1.29,0.21) -- (1.39,0.18) -- (1.50,0.15) -- (1.62,0.12) -- (1.75,0.10) -- (1.89,0.08) -- (2.04,0.06) -- (2.19,0.04) -- (2.35,0.03) -- (2.51,0.02) -- (2.67,0.01) -- (2.83,0.00) -- (3.00,0.00);
\draw[] (3.00,0.00) -- (3.17,0.00) -- (3.33,0.01) -- (3.49,0.02) -- (3.65,0.03) -- (3.81,0.04) -- (3.96,0.06) -- (4.11,0.08) -- (4.25,0.10) -- (4.38,0.12) -- (4.50,0.15) -- (4.61,0.18) -- (4.71,0.21) -- (4.80,0.24) -- (4.87,0.28) -- (4.94,0.31) -- (4.98,0.35) -- (5.01,0.38) -- (5.03,0.42) -- (5.02,0.46) -- (5.00,0.50);
\draw[] (5.00,0.50) -- (4.98,0.52) -- (4.95,0.54) -- (4.92,0.57) -- (4.89,0.59) -- (4.85,0.61) -- (4.80,0.63) -- (4.75,0.66) -- (4.70,0.68) -- (4.65,0.70) -- (4.60,0.73) -- (4.54,0.75) -- (4.48,0.78) -- (4.42,0.80) -- (4.36,0.83) -- (4.30,0.86) -- (4.24,0.88) -- (4.17,0.91) -- (4.12,0.94) -- (4.06,0.97) -- (4.00,1.00);
\draw[] (4.00,1.00) -- (3.93,1.04) -- (3.87,1.08) -- (3.80,1.12) -- (3.74,1.17) -- (3.68,1.21) -- (3.63,1.25) -- (3.58,1.30) -- (3.52,1.35) -- (3.48,1.40) -- (3.43,1.45) -- (3.38,1.50) -- (3.34,1.55) -- (3.29,1.60) -- (3.25,1.66) -- (3.21,1.71) -- (3.17,1.77) -- (3.12,1.82) -- (3.08,1.88) -- (3.04,1.94) -- (3.00,2.00);
\draw[] (3.00,2.00) -- (2.96,2.06) -- (2.92,2.12) -- (2.87,2.18) -- (2.83,2.24) -- (2.79,2.30) -- (2.74,2.36) -- (2.69,2.42) -- (2.65,2.48) -- (2.60,2.54) -- (2.55,2.59) -- (2.50,2.65) -- (2.45,2.70) -- (2.40,2.75) -- (2.35,2.80) -- (2.29,2.84) -- (2.24,2.88) -- (2.18,2.91) -- (2.12,2.95) -- (2.06,2.98) -- (2.00,3.00);
\draw[] (2.00,3.00) -- (1.96,3.01) -- (1.91,3.03) -- (1.86,3.04) -- (1.82,3.04) -- (1.77,3.05) -- (1.72,3.06) -- (1.67,3.06) -- (1.62,3.06) -- (1.57,3.06) -- (1.52,3.06) -- (1.47,3.06) -- (1.42,3.05) -- (1.37,3.05) -- (1.32,3.04) -- (1.26,3.04) -- (1.21,3.03) -- (1.16,3.02) -- (1.11,3.02) -- (1.05,3.01) -- (1.00,3.00);
\draw[dashed,] (2.00,3.50) rectangle (4.00,0.50);
\draw (0.50,3.00) node{$X$};
\draw (5.50,3.00) node{$X$};
\end{tikzpicture}
      \qeq
      \begin{tikzpicture}[baseline=(current bounding box.center),xscale=0.50,yscale=0.50,scale=0.80]
\useasboundingbox (-0.5,-0.5) rectangle (6.5,4.5);
\draw[,,] (1.00,3.00) -- (5.00,3.00);
\draw (0.50,3.00) node{$X$};
\draw (5.50,3.00) node{$X$};
\end{tikzpicture}

    \end{center}
  \end{enumerate}

\begin{proposition}
  The construction of Definition~\ref{def:net-constr} induces a trace
  on~$\Net_\Sigma$ and on~$\sNet_\Sigma$.
\end{proposition}

The category~$\sNet_\sigma$ is a traced cartesian category that we call a
\emph{fixpoint category} in the following. Interestingly, it is actually
characterized by this structure in the sense that it is the free fixpoint
category containing a $\Sigma$-object.

\begin{definition}[$\Sigma$-object]
  Given a signature~$\Sigma$, a \emph{$\Sigma$-object} in a monoidal
  category~$\C$ consists of an object~$A$ together with a morphism
  $f_\alpha:\otimes^{\sigma(\alpha)}A\to\otimes^{\tau(\alpha)}A$ for every
  symbol~$\alpha\in\Sigma$, called the \emph{interpretation} of~$\alpha$,
  where~$\otimes^nA$ denotes the tensor product of~$n$ copies of the
  object~$A$. A \emph{$\Sigma$-morphism} between two $\Sigma$\nbd{}objects
  $(A,f_\alpha)$ and~$(B,g_\alpha)$ consists of a morphism~$h:A\to B$ such that
  for every~$\alpha\in\Sigma$, $(\otimes^{\tau(\alpha)}h)\circ
  f_\alpha=g_\alpha\circ(\otimes^{\sigma(\alpha)}h)$.
\end{definition}

\begin{theorem}
  \label{thm:free-sigma}
  The category~$\sNet_\Sigma$ is the free category containing a $\Sigma$-object
  in the sense that for every fixpoint category~$\C$, the
  category~$\Mod(\Sigma,\C)$ of $\Sigma$-objects and $\Sigma$-morphisms is
  equivalent to the category~$\Fix(\sNet_\Sigma,\C)$ whose objects are fixpoint
  functors (preserving cartesian product and trace) and morphisms are monoidal
  natural transformations.
\end{theorem}
\begin{proof}
  The category~$\sNet_\Sigma$ contains a $\Sigma$-object whose underlying object
  is~$1$ and the interpretation of a symbol~$\alpha$ with~$\sigma(\alpha)=m$
  and~$\tau(\alpha)=n$ is the net $N:m\to n$ such that~$P=\intset{m+n}$,
  $O=\{x\}$, $\lambda(x)=\alpha$, $s(x,i)=i$, $s(k)=m+k$, $t(x,i)=m+i$,
  $t(k)=k$:
  
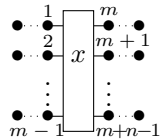
\vspace{-3ex}
  \begin{center}
  \begin{tikzpicture}[baseline=(current bounding box.center),xscale=0.50,yscale=0.50,scale=0.80]
\useasboundingbox (-0.5,-0.5) rectangle (4.5,5.5);
\draw[,,dotted] (0.00,4.00) -- (0.90,4.00);
\draw[,,dotted] (4.00,4.00) -- (3.10,4.00);
\draw[,,dotted] (0.00,3.00) -- (0.90,3.00);
\draw[,,] (1.00,4.00) -- (2.00,4.00);
\draw[,,] (1.00,3.00) -- (2.00,3.00);
\draw[,,] (1.00,1.00) -- (2.00,1.00);
\draw[,,] (3.00,4.00) -- (2.00,4.00);
\draw[,,] (3.00,3.00) -- (2.00,3.00);
\draw[,,] (3.00,1.00) -- (2.00,1.00);
\draw[,,dotted] (4.00,3.00) -- (3.10,3.00);
\draw[,,dotted] (0.00,1.00) -- (0.90,1.00);
\draw[,,dotted] (4.00,1.00) -- (3.10,1.00);
\filldraw[,fill=white] (1.50,0.50) rectangle (2.50,4.50);
\draw (1.00,4.50) node{\scriptsize{$1$}};
\draw (3.00,4.50) node{\scriptsize{$m$}};
\draw (0.00,4.00) node{$\bullet$};
\draw (1.00,3.50) node{\scriptsize{$2$}};
\draw (1.00,4.00) node{$\bullet$};
\draw (3.00,4.00) node{$\bullet$};
\draw (3.00,3.50) node{\scriptsize{~~~~$m+1$}};
\draw (4.00,4.00) node{$\bullet$};
\draw (0.00,3.00) node{$\bullet$};
\draw (1.00,3.00) node{$\bullet$};
\draw (2.00,3.00) node{$x$};
\draw (3.00,3.00) node{$\bullet$};
\draw (4.00,3.00) node{$\bullet$};
\draw (1.00,2.00) node{$\vdots$};
\draw (3.00,2.00) node{$\vdots$};
\draw (0.00,1.00) node{$\bullet$};
\draw (1.00,1.00) node{$\bullet$};
\draw (3.00,1.00) node{$\bullet$};
\draw (4.00,1.00) node{$\bullet$};
\draw (1.00,0.50) node{\scriptsize{$m-1$~~~~}};
\draw (4.00,0.50) node{\scriptsize{$m\!\!+\!\!n\!\!-\!\!1$~~~~}};
\end{tikzpicture}
  \vspace{-2ex}
  \end{center}
  A construction of the free traced symmetric monoidal category containing a
  $\Sigma$-object was provided in~\cite{abramsky:traced-compact-closed} and
  reformulated in~\cite{hasegawa2008finite} using a variant of the nets defined
  here, that we call \emph{traced nets}. It is easy to see that we recover
  traced nets by restricting~$\sNet_\Sigma$ to the nets such that the source
  function~$s$ is a bijective function. We thus have to show that our category
  of nets is the free category over the category of traced nets. Recall that a
  cocommutative comonoid~$(M,\delta,\varepsilon)$ in a symmetric monoidal
  category consists of an object~$M$ together with two morphisms $\delta:M\to
  M\otimes M$ (called \emph{duplication}) and~$\varepsilon:M\to I$ (called
  \emph{erasure}), which are such that
  $(\delta\otimes\id_I)\circ\delta=(\id_I\otimes\delta)\circ\delta$,
  $(\varepsilon\otimes\id_I)\circ\delta=\delta=(\id_I\otimes\varepsilon)\circ\delta$
  and $\gamma_{M,M}\circ\delta=\delta$. Now, it has been
  shown~\cite{burroni1993higher} that the category whose objects are integers
  and whose morphisms~$f:m\to n$ are functions~$f:\intset{m}\to\intset{n}$ is
  the free monoidal symmetric monoidal category containing a commutative monoid,
  and that the free cartesian category over a symmetric monoidal category is
  obtained by freely adding a natural structure of cocommutative comonoids over
  each object: precisely, this means that each object~$M$ is equipped with a
  structure~$(M,\delta_M,\varepsilon_M)$ of cocommutative comonoid and these are
  natural in the sense that for every morphism~$f:M\to N$, $\delta_N\circ
  f=(f\otimes f)\circ\delta_M$ and~$\varepsilon_N\circ f=\varepsilon_M$. From
  this it can be deduced that going from nets with bijective~$s$ to nets with
  any function as~$s$, and quotienting by se-equivalence, amounts to construct
  the free cartesian category over the category of traced nets. Namely, allowing
  any function equips the object~$1$ with a structure of comonoid with the
  duplication~$\delta_1$ being the net~$N_{\delta_1}:1\to 2$ such
  that~$P=\{p\}$, $O=\emptyset$, $s(k)=p$ and~$t(k)=p$ and the
  duplication~$\varepsilon_1$ being the net~$N_{\varepsilon_1}:1\to 0$ such that
  $P=\{p\}$, $O=\emptyset$ and~$t(k)=p$:
  \begin{center}
  $N_{\delta_1}$$\qeq$\begin{tikzpicture}[baseline=(current bounding box.center),xscale=0.50,yscale=0.50,scale=0.80]
\useasboundingbox (-0.5,-0.5) rectangle (2.5,2.5);
\draw[,,dotted] (2.00,2.00) -- (1.00,1.10);
\draw[,,dotted] (0.00,1.00) -- (0.90,1.00);
\draw[,,dotted] (2.00,0.00) -- (1.00,0.90);
\draw (2.00,2.00) node{$\bullet$};
\draw (0.00,1.00) node{$\bullet$};
\draw (1.00,1.00) node{$\bullet$};
\draw (2.00,0.00) node{$\bullet$};
\end{tikzpicture}
  \qquad\qquad\qquad
  $N_{\varepsilon_1}$$\qeq$\begin{tikzpicture}[baseline=(current bounding box.center),xscale=0.50,yscale=0.50,scale=0.80]
\useasboundingbox (-0.5,-0.5) rectangle (1.5,2.5);
\draw[,,dotted] (0.00,1.00) -- (0.90,1.00);
\draw (0.00,1.00) node{$\bullet$};
\draw (1.00,1.00) node{$\bullet$};
\end{tikzpicture}

  \end{center}
  (and every object can be equipped with a structure of cocommutative comonoid
  in a similar way). Quotienting by se-equivalence amounts to impose that the
  structures of cocommutative comonoid on the objects are natural.
\end{proof}



\subsection{Models of nets}
\label{sec:net-models}
The properties of fixpoint categories have been studied in details by Hasegawa
and Hyland~\cite{hasegawa1997recursion}. In particular, they have shown that a
cartesian category~$\C$ is traced if and only if it contains a fixpoint operator
satisfying suitable axioms (these are sometimes called \emph{Conway fixpoint
  operators}). For instance, the category of Scott domains recalled below admits
such a fixpoint and is therefore a fixpoint category
thus providing a natural semantics for nets.

A \emph{directed complete partial order} (or \emph{dcpo}) is a poset~$(D,\leq)$
such that every directed subset~$X$ has a supremum~$\bigvee X$ and a
\emph{complete partial order} (or \emph{cpo}) is a dcpo with a least
element~$\bot$~\cite{Abramsky:domain-theory-handbook, amadio1998domains,
  davey2002introduction}. A function~$f:A\to B$ between two dcpo is
\emph{Scott-continuous} when it preserves supremums. By the Kleene fixpoint
theorem, every Scott-continuous function \hbox{$f:X\to X$} admits a least
fixpoint~$\fix[X]f$ defined by~$\fix[X]{f}=\bigvee_{n\in\N}f^n(\bot_X)$,
where~$f^n$ denotes the composite of $n$ instances of~$f$. Suppose given a
function \hbox{$f:A\times X\to B\times X$}. We write \hbox{$\pi_B:B\times X\to
  B$} and \hbox{$\pi_X:B\times X\to X$} for the canonical
projections. Given~$a\in A$, we write~$f_a=x\mapsto f(a,x)$ for the partial
application of~$f$ to~$a$. A trace can defined on~$f$ by
\begin{equation}
  \label{eq:fix-trace}
  \Tr{A,B}X(f)\qeq a\qmapsto\pi_B(\fix[B\times X]{f_a\circ\pi_X})
\end{equation}
and this function can be shown to be Scott-continuous.

\begin{proposition}
  The category~$\Cpo$ of cpo and Scott-continuous functions is a fixpoint category
  with~\eqref{eq:fix-trace} as trace.
\end{proposition}

By Theorem~\ref{thm:free-sigma}, any $\Sigma$-object in~$\Cpo$ canonically
induces a functor \hbox{$F:\Net_\Sigma\to\Cpo$} which associates to every net
its \emph{domain semantics}: once we have interpreted each symbol as a
Scott-continuous function (by fixing a $\Sigma$-object), the interpretation of each
network is also fixed. In particular, when the $\Sigma$\nbd{}object is the
domain~$R^\infty$ of $R$-valued streams (for some set~$R$), we recover the usual
Kahn semantics~\cite{kahn:semantics-parallel} of nets: $R^\infty$ is the domain
whose elements are finite or infinite sequences (called \emph{streams}) of
elements of~$R$, ordered by inclusion. The intuition here is that time is
discrete (because we only consider the times where some information is passed on
a wire as \emph{instants}) and the elements of the domain are the sequences of
values passed on wires at the various instants.

\begin{example}
  Consider a signature~$\Sigma$ containing two symbols~$+:2\to 1$
  and~$\iota:1\to 1$. We consider the $\Sigma$-object~$\R^\infty$ in~$\Cpo$ such
  that the interpretation of~$+$ is the Scott-continuous
\end{example}

\vspace{-1.8ex}
\begin{wrapfigure}{r}{5cm}
  \vspace{-5ex}
  \begin{tikzpicture}[baseline=(current bounding box.center),xscale=0.50,yscale=0.50,scale=0.80]
\useasboundingbox (-0.5,-0.5) rectangle (12.5,4.5);
\draw[,,dotted] (0.00,3.00) -- (1.90,3.00);
\draw[,,] (2.00,3.00) -- (4.00,3.00);
\draw[] (9.00,1.00) -- (9.05,1.01) -- (9.09,1.01) -- (9.14,1.02) -- (9.18,1.02) -- (9.23,1.02) -- (9.27,1.03) -- (9.32,1.03) -- (9.37,1.04) -- (9.42,1.04) -- (9.47,1.04) -- (9.51,1.04) -- (9.56,1.04) -- (9.62,1.04) -- (9.67,1.04) -- (9.72,1.03) -- (9.77,1.03) -- (9.83,1.02) -- (9.88,1.02) -- (9.94,1.01) -- (10.00,1.00);
\draw[] (10.00,1.00) -- (10.07,0.99) -- (10.14,0.97) -- (10.20,0.96) -- (10.27,0.94) -- (10.34,0.92) -- (10.41,0.90) -- (10.48,0.87) -- (10.54,0.85) -- (10.61,0.83) -- (10.67,0.80) -- (10.72,0.77) -- (10.78,0.74) -- (10.82,0.72) -- (10.87,0.69) -- (10.91,0.66) -- (10.94,0.63) -- (10.96,0.59) -- (10.98,0.56) -- (11.00,0.53) -- (11.00,0.50);
\draw[] (11.00,0.50) -- (11.00,0.47) -- (10.98,0.44) -- (10.96,0.41) -- (10.94,0.37) -- (10.91,0.34) -- (10.87,0.31) -- (10.82,0.28) -- (10.77,0.25) -- (10.72,0.23) -- (10.66,0.20) -- (10.60,0.17) -- (10.54,0.15) -- (10.48,0.12) -- (10.41,0.10) -- (10.34,0.08) -- (10.27,0.06) -- (10.20,0.04) -- (10.13,0.03) -- (10.07,0.01) -- (10.00,0.00);
\draw[] (10.00,0.00) -- (9.94,-0.01) -- (9.88,-0.02) -- (9.83,-0.02) -- (9.77,-0.03) -- (9.72,-0.03) -- (9.67,-0.04) -- (9.62,-0.04) -- (9.57,-0.04) -- (9.52,-0.04) -- (9.47,-0.04) -- (9.42,-0.03) -- (9.37,-0.03) -- (9.32,-0.03) -- (9.28,-0.02) -- (9.23,-0.02) -- (9.18,-0.02) -- (9.14,-0.01) -- (9.09,-0.01) -- (9.05,-0.00) -- (9.00,0.00);
\draw[] (9.00,0.00) -- (8.81,0.01) -- (8.63,0.03) -- (8.44,0.03) -- (8.24,0.04) -- (8.05,0.04) -- (7.85,0.04) -- (7.66,0.04) -- (7.46,0.04) -- (7.26,0.04) -- (7.05,0.04) -- (6.85,0.03) -- (6.65,0.03) -- (6.44,0.02) -- (6.24,0.01) -- (6.03,0.01) -- (5.83,0.01) -- (5.62,0.00) -- (5.41,-0.00) -- (5.21,-0.00) -- (5.00,0.00);
\draw[] (5.00,0.00) -- (4.90,0.00) -- (4.79,0.00) -- (4.69,0.00) -- (4.59,0.01) -- (4.48,0.01) -- (4.38,0.01) -- (4.28,0.01) -- (4.18,0.02) -- (4.08,0.02) -- (3.98,0.02) -- (3.88,0.02) -- (3.78,0.02) -- (3.68,0.02) -- (3.58,0.02) -- (3.48,0.02) -- (3.38,0.02) -- (3.28,0.02) -- (3.19,0.01) -- (3.09,0.01) -- (3.00,0.00);
\draw[] (3.00,0.00) -- (2.95,-0.00) -- (2.91,-0.01) -- (2.86,-0.01) -- (2.81,-0.02) -- (2.77,-0.02) -- (2.72,-0.02) -- (2.67,-0.03) -- (2.63,-0.03) -- (2.58,-0.03) -- (2.53,-0.03) -- (2.48,-0.04) -- (2.43,-0.04) -- (2.38,-0.04) -- (2.33,-0.03) -- (2.28,-0.03) -- (2.22,-0.03) -- (2.17,-0.02) -- (2.11,-0.02) -- (2.06,-0.01) -- (2.00,0.00);
\draw[] (2.00,0.00) -- (1.93,0.01) -- (1.87,0.03) -- (1.80,0.04) -- (1.73,0.06) -- (1.66,0.08) -- (1.59,0.10) -- (1.53,0.13) -- (1.46,0.15) -- (1.40,0.18) -- (1.34,0.20) -- (1.28,0.23) -- (1.23,0.26) -- (1.18,0.29) -- (1.14,0.32) -- (1.10,0.35) -- (1.06,0.38) -- (1.04,0.41) -- (1.02,0.44) -- (1.01,0.47) -- (1.00,0.50);
\draw[] (1.00,0.50) -- (1.00,0.53) -- (1.01,0.56) -- (1.03,0.59) -- (1.06,0.62) -- (1.09,0.65) -- (1.13,0.68) -- (1.17,0.71) -- (1.22,0.73) -- (1.27,0.76) -- (1.32,0.79) -- (1.38,0.81) -- (1.45,0.84) -- (1.51,0.86) -- (1.58,0.88) -- (1.65,0.91) -- (1.72,0.93) -- (1.79,0.95) -- (1.86,0.97) -- (1.93,0.98) -- (2.00,1.00);
\draw[] (2.00,1.00) -- (2.12,1.03) -- (2.24,1.05) -- (2.36,1.07) -- (2.47,1.08) -- (2.58,1.09) -- (2.69,1.10) -- (2.79,1.10) -- (2.89,1.11) -- (2.99,1.11) -- (3.09,1.10) -- (3.19,1.10) -- (3.28,1.09) -- (3.38,1.09) -- (3.47,1.08) -- (3.56,1.07) -- (3.65,1.05) -- (3.74,1.04) -- (3.82,1.03) -- (3.91,1.01) -- (4.00,1.00);
\draw[,,] (6.00,2.00) -- (4.00,2.00);
\draw[,,dotted] (12.00,2.00) -- (6.05,2.00);
\draw[] (9.00,1.00) -- (8.95,1.00) -- (8.90,0.99) -- (8.85,0.99) -- (8.80,0.99) -- (8.75,0.98) -- (8.70,0.98) -- (8.64,0.98) -- (8.59,0.98) -- (8.54,0.97) -- (8.49,0.97) -- (8.44,0.97) -- (8.39,0.97) -- (8.34,0.97) -- (8.29,0.97) -- (8.24,0.98) -- (8.19,0.98) -- (8.15,0.98) -- (8.10,0.99) -- (8.05,0.99) -- (8.00,1.00);
\draw[] (8.00,1.00) -- (7.90,1.02) -- (7.80,1.04) -- (7.70,1.07) -- (7.60,1.10) -- (7.50,1.14) -- (7.40,1.18) -- (7.30,1.22) -- (7.21,1.27) -- (7.11,1.32) -- (7.02,1.37) -- (6.92,1.43) -- (6.83,1.49) -- (6.74,1.55) -- (6.65,1.61) -- (6.56,1.67) -- (6.46,1.73) -- (6.37,1.80) -- (6.28,1.87) -- (6.19,1.93) -- (6.10,2.00);
\filldraw[,fill=white] (3.50,0.50) rectangle (4.50,3.50);
\filldraw[,fill=white] (8.50,0.50) rectangle (9.50,1.50);
\draw (0.00,3.00) node{$\bullet$};
\draw (2.00,3.00) node{$\bullet$};
\draw (4.00,2.00) node{$+$};
\draw (6.00,2.00) node{$\bullet$};
\draw (12.00,2.00) node{$\bullet$};
\draw (2.00,1.00) node{$\bullet$};
\draw (9.00,1.00) node{$\iota$};
\end{tikzpicture}
  \vspace{-5ex}
\end{wrapfigure}
\noindent
function~$\R^\infty\times\R^\infty\to\R^\infty$ such that the image of two
streams of same length is their pointwise addition and the interpretation
of~$\iota$ is the function~$\R^\infty\to\R^\infty$ which prepends a~$0$ to
streams. The interpretation of the net on the right is the function which
returns the stream whose $n$-th value is the sum of the $n+1$ first values of
the stream.
\vspace{2ex}






An element of the Kahn domain can be considered as a partial function~$s:\N\to
R$ whose domain of definition is an initial segment of~$\N$ (the integers
in~$\N$ represent the instants of the time). Our goal here is to consider a
semantics where \emph{time is continuous}, \ie streams are generalized to
partial functions~$s:\R^+\to R$ defined on an initial segment of~$\R^+$ and to
relate it to the Kahn semantics. In order to build bridge between the two
models, the intuition here is that continuous time can be considered as
``discrete'' if we allow ourselves to consider \emph{infinitesimals}: the time
in~$\R^+$ can namely be thought as a sequence of instants $0, \d t, 2\d t, 3\d
t,\ldots$ where~$\d t$ is an infinitesimal, thus enabling us to extend
techniques developed for Kahn networks to continuous time semantics. Moreover,
many operations of common use in analysis can be simply formulated by an
appropriate net with the continuous time semantics. For instance, the derivative
of a function whose definition can be formulated as~$f'(t)=(f(t)-f(t-\d t))/\d
t$ can be implemented by a net of the form~\eqref{eq:net-der}
which directly translates to nets the above formula. The rest of the paper is
devoted to explaining and formalizing these ideas by using of non-standard
analysis which allows us to rigorously make sense of the notion of
infinitesimal.

\section{A non-standard semantics for Kahn networks in continuous time}
\subsection{Hyperreals}
\label{sec:hyperreals}
We give here a brief introduction to non-standard analysis and refer the reader
to textbooks~\cite{goldblatt_lectureshyperreals_1998, robinson-book, petry1996balade} for
details. The motivation underlying the construction of hyperreals is to extend
the field~$\R$ of real numbers into a field~$\ns\R$ in which one can give a
meaning to the notions of infinitesimal and infinite numbers. Basically,
hyperreal numbers are defined as countable sequences~$(x_i)_{i\in\N}$ of real
numbers, the sequences converging towards~$0$ representing infinitesimals. Any
real~$x$ can be seen as the hyperreal which is the constant sequence whose
elements are equal to~$x$, moreover the usual operations are extended pointwise
to hyperreals, \eg the multiplication is defined by $(x_i)\times(y_i)=(x_i\times
y_i)$. In order for suitable axioms to be satisfied (for instance every non-null
hyperreal should have an inverse) one has to consider equivalence classes of
such sequences; in particular, any two sequences which only differ on a finite
number of values should be equivalent.
The starting point of non-standard analysis is the fact that a suitable
equivalence relation can be defined from an ultrafilter:

\begin{definition}[Ultrafilter]
  A \emph{filter} on a set~$I$ is a collection~$\F$ of sets which is closed
  under intersection and under supersets (\ie if $U\subseteq V\subseteq I$
  and~$U\in\F$ then~$V\in\F$). A filter is \emph{proper} when
  \hbox{$\emptyset\not\in\F$}. An \emph{ultrafilter} on~$I$ is a proper filter
  such that for every~$U\subseteq I$, either~$U\in\F$ or \hbox{$I\setminus
    U\in\F$}.
  An ultrafilter~$\F$ is \emph{principal} when there exists~$i\in I$ such that
  \hbox{$\F=\setof{U\subseteq I}{i\in U}$}, and \emph{non-principal} otherwise.
\end{definition}


\noindent
Assuming Zorn's lemma (which is equivalent to the axiom of choice), it can be
shown that

\begin{proposition}
  Any infinite set~$I$ has a non-principal ultrafilter on it.
\end{proposition}

We now fix such an ultrafilter~$\F$ on~$\N$ whose elements are called
\emph{large sets}. The fact that~$\F$ is non-principal implies that it does not
contain any finite subset of~$\N$: the construction of the ultrafilter can thus
be thought as a way of constructing a set, starting from all cofinite sets, and
coherently adding either~$I$ or its complement for every set $I\subseteq\N$
which is neither finite nor cofinite. We define an equivalence relation~$\equiv$
on countable sequences of reals by~$(x_i)\equiv(y_i)$ when
\hbox{$\setof{i\in\N}{x_i=y_i}\in\F$} and denote by $\nsclass{x_i}$ the
equivalence class of~$(x_i)$.

\begin{definition}[Hyperreals]
  The set~$\ns\R$ of \emph{hyperreals} is the set of equivalence classes (\wrt
  $\equiv$) of countable sequences of reals.
\end{definition}

\noindent
The set~$\ns\N$ of \emph{hyperintegers} is defined similarly and there is a
canonical inclusion~$\ns\N\subseteq\ns\R$.

Any countable sequence~$(x_i)$ of reals induces an hyperreal~$\nsclass{x_i}$,
and in particular a real~$r$ can be seen as an hyperreal~$\ns r=\nsclass{r}$ by
considering the equivalence class of the constant sequence whose elements are
equal to~$r$ (we sometimes leave this conversion implicit). Similarly, a
countable sequence~$(X_i)$ of subsets of~$\R$ induces a set~$\nsclass{X_i}$ of
hyperreals defined as the set of~$\nsclass{x_i}\in\nsreal$ such that
$\setof{i\in\N}{x_i\in X_i}\in\F$. A subset of~$\nsreal$ is an \emph{internal
  set} if it can be obtained this way, in particular any set~$X\subseteq\R$
induces an internal set~$\ns X=\nsclass{X}$, associated to the constant sequence
(for instance $\nsclass\R=\ns\R$). Similarly, a countable sequence of
functions~$(f_i:A_i\to B_i)$, where the~$A_i$ and~$B_i$ are subsets of~$\R$,
extends to a function~$\nsclass{f_i}:\nsclass{A_i}\to\nsclass{B_i}$, defined
on~$\nsclass{x_i}\in\nsclass{A_i}$ by
$\nsclass{f_i}(\nsclass{x_i})=\nsclass{\overline{f_i}(x_i)}$
where~$\overline{f_i}(x_i)=f_i(x_i)$ if~$x_i\in A_i$ and~$\overline{f_i}(x_i)=0$
otherwise. Such a function is called an \emph{internal function}. Any
real-valued function~$f:A\to B$ may be seen as an internal function~$\ns
f=\nsclass{f}:\nsclass{A}\to\nsclass{B}$. The notion of \emph{internal relation}
is defined similarly.

\begin{remark}
  As we explain in Section~\ref{sec:transfer}, it is important to keep in mind
  that not every set~$X\subseteq\ns\R$ (or function, or relation) is internal.
\end{remark}

\noindent
Notice that in the above definition of an internal function, we have used~$0$ as
``default value'' for the functions~$f_i$ on the elements~$x_i\not\in A_i$. This
could be avoided by choosing a suitable representative in the equivalence
class~$\nsclass{x_i}$:

\begin{lemma}
  \label{lemma:internal-element}
  Given an element~$x$ of an internal set~$\nsclass{A_i}$, there exists a
  sequence~$(y_i)$, such that~$y_i\in A_i$ for every index~$i$,
  satisfying~$\nsclass{y_i}=x$.
\end{lemma}

In the way described above, all the usual operations on reals extend to
hyperreals (and similarly for hyperintegers).
For instance, the absolute value of an hyperreal~$\hr{x}=\nsclass{x_i}$ is
defined by~$|\hr{x}|=\nsclass{|x_i|}$. An hyperreal $\hr{x}$ of $\nsreal$ is
\emph{infinitesimal} whenever $|\hr{x}|<r$ for every real $r>0$, and
\emph{unlimited} if~$r<|\hr{x}|$ for every real~$r\in\R$. Given a
hyperreal~$\hr{x}$ which is not unlimited, there exists a unique real~$y$ such
that~$\hr{x}-y$ is infinitesimal: this real is called the \emph{standard part}
of~$\hr{x}$ and denoted by~$\st(\hr{x})$. We define an equivalence
relation~$\approx$ on hyperreals by $\hr{x}\approx\hr{y}$ whenever
$\st(\hr{x}-\hr{y})=0$.

\begin{remark}
  The existence of a standard part might be surprising at first: for instance,
  given the sequence~$x_i$ such that~$x_i=0$ if $i$ is even and $x_i=1$
  otherwise, what should be the standard part of~$\nsclass{x_i}$? The result is
  given by the chosen ultrafilter~$\F$: if the set~$I$ of even integers is
  in~$\F$ then~$\st(\nsclass{x_i})=0$, otherwise the set~$\N\setminus I$ of odd
  integers is in~$\F$ and~$\st(\nsclass{x_i})=1$.
\end{remark}

\begin{remark}
  The method used to construct~$\ns\R$ is an instance of a very general
  construction of ultraproducts in model theory, which can be used to define a
  non-standard model~$\ns{\mathbb{X}}$ from any
  model~$\mathbb{X}$~\cite{los1955algebraic, robinson-book,
    goldblatt_lectureshyperreals_1998}. In particular, given sets~$\mathbb{X}$
  and~$\mathbb{Y}$, this construction applied to the
  set~$\mathbb{Y}^{\mathbb{X}}$ of functions from~$\mathbb{X}$ to~$\mathbb{Y}$
  constructs the set~$\ns{(\mathbb{Y}^{\mathbb{X}})}$ of internal functions
  from~$\ns{\mathbb{X}}$ to~$\ns{\mathbb{Y}}$.
\end{remark}

\subsubsection{The transfer principle}
\label{sec:transfer}
A fundamental tool in non-standard analysis is the \emph{transfer principle},
which follows from Łoś theorem~\cite{los1955algebraic}. Informally, this
principle can be formulated as follows

\begin{proposition}[Transfer principle]
  \label{prop:transfer}
  A first-order formula~$\varphi$ is satisfied in~$\R$ if and only if it is
  satisfied in~$\ns\R$, if we assume that all the sets, functions and relations
  involved in the formula are internal.
\end{proposition}

\noindent
A similar theorem can be formulated for~$\ns\N$. Many constructions of standard
analysis can thus be transferred to non-standard analysis. For instance, the
sets~$\ns\N$ and~$\ns\R$ are, respectively, a ring and a field and both are totally
ordered.

\begin{remark}
  The assumption that we consider only internal objects is very important. For
  instance the formula $((\forall x\in A.\ x\in\N)\land(\exists x\in\N.\ x\in
  A))\Rightarrow(\exists x\in A.\forall y\in A.\ x\leq y)$ is true in~$\N$:
  every non-empty subset~$A$ of~$\N$ admits a smallest element. From this, we
  can deduce by transfer that every non-empty \emph{internal} subset of~$\ns\N$
  admits a smallest element. However, this property does not hold for every
  subset of~$\ns\N$: for instance, $\ns\N\setminus\N$ does not have a smallest
  element since it can be shown to be closed under predecessor, and is thus not
  internal. Likewise a subset of~$\ns\N$ (\resp $\ns\R$) is internal if and only
  if it is finite.
\end{remark}

\subsubsection{Non-standard analysis}
\label{sec:ns-analysis}
One of the most interesting property of hyperreals is that it allows one to
rigorously consider infinitesimals and thus formalize in an elegant way many of
the tools in common use in standard analysis. We give below some of these
reformulations which will be of use afterward.

\begin{proposition}[Continuity]
  \label{prop:continuity}
  A function~$f:\R\to\R$ is \emph{continuous} at~$x$ when for
  every~$\hr{y}\in\ns\R$ such that $\hr{y}\approx x$, we have $\ns
  f(\hr{y})\approx f(x)$. Otherwise said, $f$ is continuous at~$x$ when for
  every infinitesimal~$\delta\approx 0$, there exists an
  infinitesimal~$\varepsilon\approx 0$ such that~$\ns f(x+\delta)=
  f(x)+\varepsilon$.
\end{proposition}

\begin{proposition}[Differentiation]
  \label{prop:differentiation}
  A function~$f:\R\to\R$ admits~$y\in\R$ as \emph{derivative} at $x\in\R$ when
  for every non-null infinitesimal~$\delta\approx 0$, we have $(\ns
  f(x+\delta)-f(x))/\delta\approx y$. Furthermore, if~$f$ is continuously
  differentiable on $\R$, then for any two non-null distinct infinitesimals
  $\delta$ and $\varepsilon$, and for any $x\in\R$, we have
  \[
  f'(x)\qeq\st\left(\frac{\ns f(x+\delta)- \ns{f}(x+\varepsilon)}{\delta-\varepsilon}\right)
  \]
\end{proposition}

Given a continuous function~$f:\R\to\R$, its integral on an interval~$[a,b]$ is
defined by
\[
\int_a^bf(x)\d x
\qeq
\lim_{n\to\infty}\pa{\sum_{k=0}^{n-1}f\pa{a+\frac kn(b-a)}\frac 1n}
\]
Notice that each sum makes sense because it is finite since it is indexed over
the finite set~$\setof{k\in\N}{0\leq k<n}$. This notion of finite set can be
generalized to internal sets as follows: an internal set~$A=\nsclass{A_i}$ is
\emph{hyperfinite} if almost all the~$A_i$ are finite, \ie
$\setof{i\in\N}{\text{$A_i$ is finite}}\in\F$. By an argument similar to
Lemma~\ref{lemma:internal-element}, we can always suppose that all the~$A_i$ are
finite by choosing a suitable sequence of finite sets~$B_i$ such
that~$\nsclass{B_i}=\nsclass{A_i}$. Given such an internal set and an internal
function~$\nsclass{f_i}$, we define
$\sum_{\nsclass{x_i}\in\nsclass{A_i}}\nsclass{f_i}(\nsclass{x_i})=\nsclass{\sum_{x_i\in A_i} f_i(x_i)}$.

\begin{proposition}[Integration]
  \label{prop:integral}
  Given a function~$f:\R\to\R$ which is continuous on an interval~$[a,b]$,
  excepting possibly a finite number of points of discontinuity, we have
  \begin{equation}
    \label{eq:integral}
    \int_a^bf(x)\d x
    \qeq
    \st\pa{\sum_{x\in N}\ns f(a+x(b-a))\delta}
  \end{equation}
  where~$\delta=1/\hr{n}$ for some unlimited~$\hr{n}=\nsclass{n_i}\in\ns\N$
  and~$N$ is the hyperfinite set $N=\nsclass{N_i}\subseteq\ns\R$
  with~$N_i=\setof{k_i/n_i\in\R}{k_i\in\N, 0\leq k_i<n_i}$. In particular, the
  result does not depend on the choice of the unlimited
  hyperinteger~$\hr{n}\in\ns\N$.
\end{proposition}

\noindent
The notion defined above corresponds to the Riemann integral. More refined notions
(such as the Lebesgue integral) can also be adapted to the non-standard setting.

\subsection{Internal domains}
\label{sec:internal-dom}
In this section, we introduce the notion of internal domain, which we
use to define a non-standard denotational semantics for process networks.
Given a totally ordered set~$T$ and a set~$R$, we write $\sdom TR$ for the set
of partial functions $s:T\to R$ defined on an initial segment of~$T$, called the
\emph{domain of definition} of~$s$. The elements of~$\sdom TR$ are called
\emph{streams}: the set~$T$ can be thought as \emph{time} and the elements
of~$R$ as the possible \emph{values} of a stream over time. Every such set can
be equipped with a partial order~$\sqsubseteq$ such that, given $r,s\in\sdom
TR$, we have $f\sqsubseteq g$ whenever the definition domain of~$r$ is included
in the definition domain of~$s$ and~$r$ and~$s$ coincide on the domain of
definition of~$r$.

\begin{proposition}
  The poset $(\sdom TR,\sqsubseteq)$ is a cpo with smallest element~$\bot$ being
  the function nowhere defined.
\end{proposition}

\begin{example}
  The Kahn domain described in Section~\ref{sec:net-models} is $\sdom\N R$.
\end{example}

\noindent
Every function~$f:R\to R$ lifts to a function~$\tilde{f}:\sdom TR\to\sdom TR$
such that, given~$s\in\sdom TR$, the domain of definition of~$\tilde{f}(s)$ is
the same as the domain of definition of~$s$ and the image of~$s$ is defined
by~$\tilde f(s)(t)=f\circ s(t)$. The function~$\tilde{f}$ is called the
\emph{lifting} of~$f$ to~$\sdom TR$. It is easy to show that every such lifting
is Scott-continuous.

In the following, we will be interested in modeling nets operating in a time
which varies continuously. We thus introduce the following domain in order to
model the data flowing through the wires:

\begin{definition}[Continuous-time domain]
  The \emph{continuous-time domain} is the complete partial order
  $\CT=\sdom{\R^+}\R$. The \emph{continuous-time domain of continuous functions}
  ~$\CCT$ is the subdomain of~$\CT$ whose elements are continuous partial
  functions~$\R^+\to\R$.
\end{definition}

\noindent
As explained in the introduction, the purpose of this paper is to explain how to
implement Scott-continuous functions over this domain using Kahn networks by
formalizing the following intuition using non-standard semantics: continuous
time can be considered as ``discrete'' where the duration between two instants
is infinitesimal. A natural candidate for this would be the
domain~$\sdom{\ns\N}{\ns\R}$. Namely, in the view of
Proposition~\ref{prop:integral}, we would like to relate a stream
$s\in\sdom{\R^+}\R$ with the stream~$\overline{s}\in\sdom{\ns\N}{\ns\R}$ defined
on~$\hr{t}\in\ns\N$ by~$\overline{s}(\hr{t})=\ns{s}(\hr{t}\delta)$, from some
infinitesimal~$\delta\in\ns\R$. It turns out that the fixpoints computed
in~$\sdom{\ns\N}{\ns\R}$ are not satisfactory.
\begin{wrapfigure}{r}{3.1cm}
  \vspace{-15ex}
  \parbox{0pt}{\begin{equation}\label{eq:net_cst}\end{equation}}
  \begin{tikzpicture}[baseline=(current bounding box.center),xscale=0.50,yscale=0.50,scale=0.80]
\useasboundingbox (-0.5,-0.5) rectangle (5.5,1.5);
\draw[] (2.00,1.00) -- (1.95,1.01) -- (1.91,1.01) -- (1.86,1.02) -- (1.81,1.02) -- (1.76,1.03) -- (1.72,1.03) -- (1.67,1.04) -- (1.62,1.04) -- (1.57,1.04) -- (1.52,1.04) -- (1.47,1.04) -- (1.42,1.04) -- (1.37,1.04) -- (1.32,1.04) -- (1.27,1.04) -- (1.22,1.03) -- (1.16,1.03) -- (1.11,1.02) -- (1.06,1.01) -- (1.00,1.00);
\draw[] (1.00,1.00) -- (0.94,0.99) -- (0.87,0.97) -- (0.81,0.95) -- (0.75,0.93) -- (0.68,0.91) -- (0.62,0.89) -- (0.56,0.86) -- (0.50,0.84) -- (0.44,0.81) -- (0.38,0.79) -- (0.33,0.76) -- (0.28,0.73) -- (0.23,0.70) -- (0.18,0.67) -- (0.14,0.64) -- (0.10,0.61) -- (0.07,0.58) -- (0.04,0.56) -- (0.02,0.53) -- (0.00,0.50);
\draw[] (0.00,0.50) -- (-0.02,0.45) -- (-0.02,0.40) -- (-0.00,0.36) -- (0.03,0.32) -- (0.08,0.28) -- (0.14,0.24) -- (0.22,0.21) -- (0.31,0.18) -- (0.41,0.15) -- (0.52,0.12) -- (0.64,0.10) -- (0.77,0.08) -- (0.90,0.06) -- (1.05,0.04) -- (1.20,0.03) -- (1.35,0.02) -- (1.51,0.01) -- (1.67,0.00) -- (1.84,0.00) -- (2.00,0.00);
\draw[] (2.00,0.00) -- (2.16,0.00) -- (2.33,0.00) -- (2.49,0.01) -- (2.65,0.02) -- (2.80,0.03) -- (2.95,0.04) -- (3.10,0.06) -- (3.23,0.08) -- (3.36,0.10) -- (3.48,0.12) -- (3.59,0.15) -- (3.69,0.18) -- (3.78,0.21) -- (3.86,0.24) -- (3.92,0.28) -- (3.97,0.32) -- (4.00,0.36) -- (4.02,0.40) -- (4.02,0.45) -- (4.00,0.50);
\draw[] (4.00,0.50) -- (3.98,0.53) -- (3.96,0.56) -- (3.93,0.58) -- (3.90,0.61) -- (3.86,0.64) -- (3.82,0.67) -- (3.77,0.70) -- (3.72,0.73) -- (3.67,0.76) -- (3.62,0.79) -- (3.56,0.81) -- (3.50,0.84) -- (3.44,0.86) -- (3.38,0.89) -- (3.32,0.91) -- (3.25,0.93) -- (3.19,0.95) -- (3.13,0.97) -- (3.06,0.99) -- (3.00,1.00);
\draw[] (3.00,1.00) -- (2.94,1.01) -- (2.89,1.02) -- (2.84,1.03) -- (2.78,1.03) -- (2.73,1.04) -- (2.68,1.04) -- (2.63,1.04) -- (2.58,1.04) -- (2.53,1.04) -- (2.48,1.04) -- (2.43,1.04) -- (2.38,1.04) -- (2.33,1.04) -- (2.28,1.03) -- (2.24,1.03) -- (2.19,1.02) -- (2.14,1.02) -- (2.09,1.01) -- (2.05,1.01) -- (2.00,1.00);
\draw[,,dotted] (5.00,1.00) -- (3.10,1.00);
\filldraw[,fill=white] (1.50,0.50) rectangle (2.50,1.50);
\draw (2.00,1.00) node{$\iota$};
\draw (3.00,1.00) node{$\bullet$};
\draw (5.00,1.00) node{$\bullet$};
\end{tikzpicture} (\theequation)
  \vspace{-5ex}
\end{wrapfigure}
For instance, consider the net on the right such that the
interpretation of the operator~$\iota$ is the
function~$\iota:\sdom{\ns\N}{\ns\R}\to\sdom{\ns\N}{\ns\R}$ such that the image
of a stream~$r$ is the stream~$s$ defined by~$s(0)=0$ and for any non-null
hyperinteger~$\hr{t}$, $s(\hr{t})=r(\hr{t}-1)$. We expect the interpretation of
this net to be the constant function equal to~$0$. However, this is not the case:
the semantics~$s$ of this net is given by the
fixpoint~$s=\fix\iota=\bigvee_{k\in\N}\iota^k(\bot)$ of~$\iota$. Given~$k\in\N$,
the domain of definition of the stream~$\iota^k(\bot)$ is the
set~$\setof{\hr{p}\in\ns\N}{0\leq\hr{p}<k}$. Therefore, given an
unlimited~$\hr{n}\in\ns\N$, $\iota^k(\bot)(\hr{n})$ is undefined for
every~$k\in\N$ and thus $\fix\iota(\hr{n})$ is undefined. Intuitively, the
induction on~$k\in\N$ defining the smallest fixpoint is not powerful enough to
reach all elements of~$\ns\N$. The cpo~$\sdom{\ns\N}{\ns\R}$ is thus not the
appropriate domain, however we explain below that internal domains are a more
suitable notion, because they support an induction principle on~$\ns\N$.

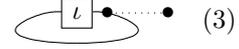
\begin{definition}[Internal cpo]
  An \emph{internal cpo} $(D,\leq)$ in a non-standard model consists of an
  internal set $D=\nsclass{D_i}$ and an internal
  relation~$\leq=\nsclass{\leq_i}$ such that for every integer~$i$,
  $(D_i,\leq_i)$ is a cpo. Similarly, an \emph{internal Scott-continuous
    function}~$f:D\to E$ between two internal cpo~$D=\nsclass{D_i}$
  and~$E=\nsclass{E_i}$ consists of an internal function~$\nsclass{f_i:D_i\to
    E_i}$ such that all the~$f_i$ are continuous. We write~$\ICpo$ for the
  category of internal cpo and internal Scott-continuous functions.
\end{definition}

\begin{remark}
  Notice that such an internal cpo~$(D,\leq)$ is not necessarily a cpo: only
  internal directed subsets are required to have a supremum. For instance
  suppose fixed an unlimited hyperinteger~$\hr{n}\in\ns\N$. The
  set~$D=\setof{\hr{k}\in\ns\N}{\hr{k}\leq\hr{n}}$ equipped with the usual total
  order is an internal cpo, but not a cpo because the (non-internal)
  subset~$\N\subseteq D$ is directed and does not have a supremum.
\end{remark}

\begin{proposition}
  \label{prop:internal-fix}
  Any internal Scott-continuous function~$f:D\to D$, where~$D$ is an internal
  cpo, admits a least fixpoint~$\fix{f}$ which
  satisfies~$\fix{f}=\bigvee\setof{f^{\hr{n}}(\bot)}{\hr{n}\in\ns\N}$. Here,
  if~$f=\nsclass{f_i}$ and~$\hr{n}=\nsclass{n_i}$, $f^{\hr{n}}$ is defined
  as~$\nsclass{f_i^{n_i}}$.
\end{proposition}


\noindent
The axioms of fixpoint categories can be formulated in first-order logic. Using
the transfer principle (Proposition~\ref{prop:transfer}), it can be shown that
the fact that~$\Cpo$ is a fixpoint category implies that

\begin{proposition}
  \label{prop:icpo-fix-point-cat}
  The category~$\ICpo$ is a fixpoint category.
\end{proposition}

In the category~$\ICpo$, we will be particularly interested in the following
domain:
\begin{definition}[Infinitesimal-time domain]
  The \emph{infinitesimal-time domain} is the internal complete partial
  order~$\IT=\ns{(\sdom\N\R)}$.
\end{definition}

\noindent
As explained in the remark in the end of Section~\ref{sec:hyperreals}, the
elements of~$\IT=\ns{(\sdom\N\R)}$ are the internal partial functions
from~$\ns\N$ to~$\ns\R$. The order~$\sqsubseteq$ on this domain is such that for
any~$r,s\in\IT$, we have $r\sqsubseteq s$ whenever the domain of definition
of~$r$ is included in the domain of definition of~$s$ and~$r$ and~$s$ coincide
on the domain of definition of~$r$.

\subsection{Comparing continuous time and infinitesimal time}
\label{sec:ct-it}
In this section, we explain how the semantics of nets in~$\IT$ can ``simulate''
operations in~$\CT$. We now suppose fixed an infinitesimal $\delta$ called
\emph{sampling period}. We define a function $S:\CT\to\IT$, called
\emph{sampling}, which to every stream~$s\in\CT$ associates the stream
$S(s)=\hr{k}\mapsto\ns{s}(\hr{k}\delta)$, and a function~$T:\IT\to\CT$, called
\emph{standardization}, which to every stream~$s$
associates~$T(s)=x\mapsto\st{(s(\floor{\ns{x}/\delta}))}$,
where~$\floor-:\ns\R\to\ns\N$ denotes the floor function, and is defined on the
biggest initial segment of~$\R^+$ for which this definition makes sense. These
functions enable us to show that~$\CCT$ (the domain of \emph{continuous}
streams) is a retract of~$\IT$. We discuss afterward the possible extensions of
this result to elements of~$\CT$.

\begin{proposition}
  \label{prop:retract}
  The restriction of the composite~$T\circ S$ to~$\CCT$ is the identity.
\end{proposition}
\begin{proof}
  Suppose given a stream~$s\in\CCT$. For any~$x\in\R^+$, the fact that~$s$ is
  continuous at~$x$ implies, by Proposition~\ref{prop:continuity}, that for
  every~$\hr{k}\in\ns\N$ such that $\hr{k}\delta\approx x$, we have
  $S(s)(\hr{k}\delta)\approx s(x)$. From this we deduce that $T(S(s))(x)=s(x)$.
\end{proof}

\begin{remark}
  The function~$T\circ S$ is generally \emph{not} the identity on~$\CT$. For
  instance, suppose that~$\delta=1/\hr{n}$, where~$\hr{n}\in\ns\N$ is unlimited,
  and consider the stream~$s\in\CT$ whose value is~$0$ everywhere except
  at~$\sqrt{2}$ where~$s(\sqrt{2})=1$. Using the transfer principle, it is easy
  to show that for every~$\hr{k}\in\ns\N$, we have
  $\hr{k}/\hr{n}\neq\sqrt{2}$. From this we can deduce that~$T\circ S(s)$ is the
  constant stream equal to~$0$.
\end{remark}

\noindent
In order to make a more convincing case of the interest of the domain~$\IT$ as a
model of nets and study further its relationship with~$\CT$, we give below some
examples of nets interpreted in~$\IT$ which implement common constructions in
analysis, and relate them to the corresponding constructions in~$\CT$
through~$S$ and~$T$. For concision, we do not detail the easy verification that
the interpretations of operators are internal Scott-continuous functions.

As a first simple example, consider the net~\eqref{eq:net_cst}. From the
characterization of the fixpoint of internal Scott-continuous functions given by
Proposition~\ref{prop:internal-fix}, it is easy to check that its semantics~$s$
in the domain~$\IT$ is the constant function (defined everywhere) as expected:
if we write~$c_0:\R^+\to\R$ for the constant function equal to~$0$, we have
$s=S(c_0)$ and~$c_0=T(s)$.

\subsubsection{Differentiation}
The \emph{differentiation} is the following net where ``$\varepsilon$'' is
interpreted as the function which drops the first element of a stream (\ie
$\varepsilon(s)(\hr{k})=s(\hr{k}+1)$), ``$-$'' is interpreted as the pointwise
difference (\ie $(-(s,t))(\hr{k})=s(\hr{k})-t(\hr{k})$), and ``$/\delta$'' is
interpreted as the pointwise division by~$\delta$ (\ie
$(s/\delta)(\hr{k})=s(\hr{k})/\delta$).
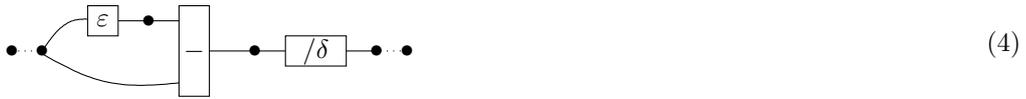
\begin{equation}
  \label{eq:net-der}
  \begin{tikzpicture}[baseline=(current bounding box.center),xscale=0.50,yscale=0.50,scale=0.80]
\useasboundingbox (-0.5,-0.5) rectangle (13.5,2.5);
\draw[] (3.00,2.00) -- (2.95,2.01) -- (2.89,2.01) -- (2.84,2.02) -- (2.79,2.03) -- (2.74,2.03) -- (2.68,2.04) -- (2.63,2.04) -- (2.58,2.05) -- (2.53,2.05) -- (2.48,2.05) -- (2.43,2.06) -- (2.38,2.06) -- (2.33,2.05) -- (2.28,2.05) -- (2.23,2.05) -- (2.18,2.04) -- (2.14,2.03) -- (2.09,2.03) -- (2.04,2.01) -- (2.00,2.00);
\draw[] (2.00,2.00) -- (1.94,1.98) -- (1.88,1.95) -- (1.82,1.92) -- (1.77,1.88) -- (1.71,1.85) -- (1.66,1.80) -- (1.60,1.76) -- (1.55,1.71) -- (1.50,1.66) -- (1.45,1.61) -- (1.41,1.56) -- (1.36,1.50) -- (1.31,1.44) -- (1.27,1.38) -- (1.22,1.32) -- (1.18,1.26) -- (1.13,1.19) -- (1.09,1.13) -- (1.04,1.06) -- (1.00,1.00);
\draw[] (6.00,2.00) -- (5.90,2.00) -- (5.80,2.00) -- (5.70,2.00) -- (5.60,2.00) -- (5.50,2.00) -- (5.40,2.00) -- (5.30,2.00) -- (5.20,2.00) -- (5.10,2.00) -- (5.00,2.00) -- (4.90,2.00) -- (4.80,2.00) -- (4.70,2.00) -- (4.60,2.00) -- (4.50,2.00) -- (4.40,2.00) -- (4.30,2.00) -- (4.20,2.00) -- (4.10,2.00) -- (4.00,2.00);
\draw[] (4.00,2.00) -- (3.95,2.00) -- (3.90,2.00) -- (3.85,2.00) -- (3.80,2.00) -- (3.75,2.00) -- (3.70,2.00) -- (3.65,2.00) -- (3.60,2.00) -- (3.55,2.00) -- (3.50,2.00) -- (3.45,2.00) -- (3.40,2.00) -- (3.35,2.00) -- (3.30,2.00) -- (3.25,2.00) -- (3.20,2.00) -- (3.15,2.00) -- (3.10,2.00) -- (3.05,2.00) -- (3.00,2.00);
\draw[,,dotted] (0.00,1.00) -- (0.90,1.00);
\draw[] (1.00,0.90) -- (1.10,0.85) -- (1.20,0.79) -- (1.29,0.74) -- (1.39,0.68) -- (1.49,0.63) -- (1.59,0.58) -- (1.69,0.53) -- (1.79,0.48) -- (1.89,0.43) -- (1.98,0.38) -- (2.08,0.33) -- (2.18,0.29) -- (2.28,0.24) -- (2.39,0.20) -- (2.49,0.16) -- (2.59,0.13) -- (2.69,0.09) -- (2.79,0.06) -- (2.90,0.03) -- (3.00,0.00);
\draw[] (3.00,0.00) -- (3.14,-0.03) -- (3.29,-0.06) -- (3.43,-0.08) -- (3.58,-0.10) -- (3.72,-0.12) -- (3.87,-0.13) -- (4.02,-0.13) -- (4.17,-0.14) -- (4.32,-0.14) -- (4.47,-0.13) -- (4.62,-0.13) -- (4.77,-0.12) -- (4.93,-0.11) -- (5.08,-0.10) -- (5.23,-0.08) -- (5.39,-0.07) -- (5.54,-0.05) -- (5.69,-0.04) -- (5.85,-0.02) -- (6.00,0.00);
\draw[] (10.00,1.00) -- (9.90,1.00) -- (9.80,1.00) -- (9.70,1.00) -- (9.60,1.00) -- (9.50,1.00) -- (9.40,1.00) -- (9.30,1.00) -- (9.20,1.00) -- (9.10,1.00) -- (9.00,1.00) -- (8.90,1.00) -- (8.80,1.00) -- (8.70,1.00) -- (8.60,1.00) -- (8.50,1.00) -- (8.40,1.00) -- (8.30,1.00) -- (8.20,1.00) -- (8.10,1.00) -- (8.00,1.00);
\draw[] (8.00,1.00) -- (7.90,1.00) -- (7.80,1.00) -- (7.70,1.00) -- (7.60,1.00) -- (7.50,1.00) -- (7.40,1.00) -- (7.30,1.00) -- (7.20,1.00) -- (7.10,1.00) -- (7.00,1.00) -- (6.90,1.00) -- (6.80,1.00) -- (6.70,1.00) -- (6.60,1.00) -- (6.50,1.00) -- (6.40,1.00) -- (6.30,1.00) -- (6.20,1.00) -- (6.10,1.00) -- (6.00,1.00);
\draw[,,] (12.00,1.00) -- (10.00,1.00);
\draw[,,dotted] (13.00,1.00) -- (12.10,1.00);
\filldraw[,fill=white] (2.50,1.50) rectangle (3.50,2.50);
\filldraw[,fill=white] (5.50,-0.50) rectangle (6.50,2.50);
\filldraw[,fill=white] (9.00,0.50) rectangle (11.00,1.50);
\draw (3.00,2.00) node{$\varepsilon$};
\draw (4.50,2.00) node{$\bullet$};
\draw (0.00,1.00) node{$\bullet$};
\draw (1.00,1.00) node{$\bullet$};
\draw (6.00,1.00) node{$-$};
\draw (8.00,1.00) node{$\bullet$};
\draw (10.00,1.00) node{$/\delta$};
\draw (12.00,1.00) node{$\bullet$};
\draw (13.00,1.00) node{$\bullet$};
\end{tikzpicture}
\end{equation}

\noindent
We write~$\varphi:\IT\to\IT$ for the semantics of the net. By
Proposition~\ref{prop:differentiation}, we have immediately:
\begin{proposition}
  For any continuously differentiable function~$s:\R^+\to\R$ in $\CCT$,
  $T(\varphi(S(s))) = s'$.
\end{proposition}
\begin{proof}
  Given~$\hr{k}\in\nsint$ such that $\st(\hr{k}\delta)\in\Rplus$,
  $\varphi(S(s))(\hr{k})={S(s)(\hr{k}+1)- S(s)(\hr{k}) \over \delta}\approx$
  $s'(\st(\hr{k}\delta))$. The second step is proved by
  Proposition~\ref{prop:differentiation}. Therefore, $T(\varphi(S(s)))=s'$.
\end{proof}

\subsubsection{Integration}
The \emph{integration} is the following net where ``$\times\delta$'' is the
pointwise multiplication of a stream by~$\delta$, ``$+$'' is the pointwise
addition of streams, and~$\iota$ prepends~$0$ to a stream (see
Section~\ref{sec:internal-dom}).
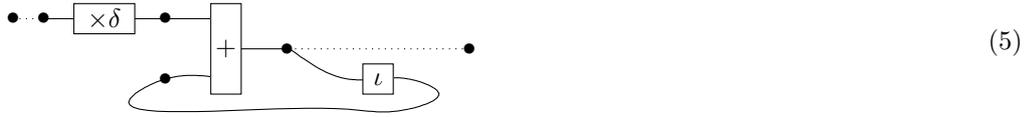
\vspace{-4ex}
\begin{equation}
  \label{eq:net-int}
  \begin{tikzpicture}[baseline=(current bounding box.center),xscale=0.50,yscale=0.50,scale=0.80]
\useasboundingbox (-0.5,-0.5) rectangle (15.5,2.5);
\draw[,,dotted] (0.00,2.00) -- (0.90,2.00);
\draw[,,] (1.00,2.00) -- (3.00,2.00);
\draw[] (7.00,2.00) -- (6.90,2.00) -- (6.80,2.00) -- (6.70,2.00) -- (6.60,2.00) -- (6.50,2.00) -- (6.40,2.00) -- (6.30,2.00) -- (6.20,2.00) -- (6.10,2.00) -- (6.00,2.00) -- (5.90,2.00) -- (5.80,2.00) -- (5.70,2.00) -- (5.60,2.00) -- (5.50,2.00) -- (5.40,2.00) -- (5.30,2.00) -- (5.20,2.00) -- (5.10,2.00) -- (5.00,2.00);
\draw[] (5.00,2.00) -- (4.90,2.00) -- (4.80,2.00) -- (4.70,2.00) -- (4.60,2.00) -- (4.50,2.00) -- (4.40,2.00) -- (4.30,2.00) -- (4.20,2.00) -- (4.10,2.00) -- (4.00,2.00) -- (3.90,2.00) -- (3.80,2.00) -- (3.70,2.00) -- (3.60,2.00) -- (3.50,2.00) -- (3.40,2.00) -- (3.30,2.00) -- (3.20,2.00) -- (3.10,2.00) -- (3.00,2.00);
\draw[] (12.00,0.00) -- (12.05,0.01) -- (12.09,0.01) -- (12.14,0.02) -- (12.18,0.02) -- (12.23,0.03) -- (12.27,0.03) -- (12.32,0.03) -- (12.37,0.04) -- (12.42,0.04) -- (12.47,0.04) -- (12.51,0.04) -- (12.56,0.04) -- (12.62,0.04) -- (12.67,0.04) -- (12.72,0.04) -- (12.77,0.03) -- (12.83,0.03) -- (12.88,0.02) -- (12.94,0.01) -- (13.00,0.00);
\draw[] (13.00,0.00) -- (13.07,-0.01) -- (13.14,-0.03) -- (13.20,-0.04) -- (13.27,-0.06) -- (13.34,-0.08) -- (13.41,-0.10) -- (13.48,-0.13) -- (13.54,-0.15) -- (13.60,-0.18) -- (13.66,-0.20) -- (13.72,-0.23) -- (13.77,-0.26) -- (13.82,-0.29) -- (13.87,-0.32) -- (13.91,-0.35) -- (13.94,-0.38) -- (13.96,-0.41) -- (13.98,-0.44) -- (14.00,-0.47) -- (14.00,-0.50);
\draw[] (14.00,-0.50) -- (14.00,-0.53) -- (13.98,-0.56) -- (13.97,-0.59) -- (13.94,-0.62) -- (13.91,-0.65) -- (13.87,-0.68) -- (13.83,-0.71) -- (13.78,-0.74) -- (13.73,-0.76) -- (13.67,-0.79) -- (13.61,-0.82) -- (13.55,-0.84) -- (13.48,-0.86) -- (13.41,-0.89) -- (13.35,-0.91) -- (13.28,-0.93) -- (13.21,-0.95) -- (13.14,-0.97) -- (13.07,-0.98) -- (13.00,-1.00);
\draw[] (13.00,-1.00) -- (12.88,-1.02) -- (12.77,-1.04) -- (12.66,-1.06) -- (12.55,-1.07) -- (12.45,-1.08) -- (12.35,-1.09) -- (12.25,-1.09) -- (12.16,-1.09) -- (12.07,-1.09) -- (11.97,-1.09) -- (11.88,-1.08) -- (11.79,-1.07) -- (11.70,-1.07) -- (11.60,-1.06) -- (11.51,-1.05) -- (11.41,-1.04) -- (11.31,-1.03) -- (11.21,-1.02) -- (11.11,-1.01) -- (11.00,-1.00);
\draw[] (11.00,-1.00) -- (10.59,-0.98) -- (10.14,-0.97) -- (9.66,-0.97) -- (9.16,-0.98) -- (8.63,-0.99) -- (8.10,-1.01) -- (7.57,-1.03) -- (7.04,-1.05) -- (6.52,-1.07) -- (6.03,-1.08) -- (5.57,-1.08) -- (5.14,-1.08) -- (4.76,-1.07) -- (4.43,-1.04) -- (4.16,-1.00) -- (3.96,-0.94) -- (3.84,-0.86) -- (3.80,-0.77) -- (3.85,-0.65) -- (4.00,-0.50);
\draw[] (4.00,-0.50) -- (4.03,-0.47) -- (4.07,-0.45) -- (4.11,-0.42) -- (4.15,-0.39) -- (4.19,-0.37) -- (4.24,-0.34) -- (4.29,-0.31) -- (4.33,-0.28) -- (4.38,-0.26) -- (4.44,-0.23) -- (4.49,-0.20) -- (4.54,-0.18) -- (4.60,-0.15) -- (4.65,-0.13) -- (4.71,-0.10) -- (4.77,-0.08) -- (4.83,-0.06) -- (4.88,-0.04) -- (4.94,-0.02) -- (5.00,0.00);
\draw[] (5.00,0.00) -- (5.10,0.03) -- (5.21,0.05) -- (5.31,0.07) -- (5.41,0.09) -- (5.51,0.10) -- (5.61,0.11) -- (5.71,0.12) -- (5.81,0.12) -- (5.91,0.12) -- (6.01,0.12) -- (6.11,0.11) -- (6.21,0.11) -- (6.31,0.10) -- (6.41,0.09) -- (6.51,0.07) -- (6.61,0.06) -- (6.71,0.05) -- (6.80,0.03) -- (6.90,0.02) -- (7.00,0.00);
\draw[,,] (9.00,1.00) -- (7.00,1.00);
\draw[,,dotted] (15.00,1.00) -- (9.05,1.00);
\draw[] (12.00,0.00) -- (11.95,-0.00) -- (11.90,-0.01) -- (11.85,-0.01) -- (11.80,-0.01) -- (11.75,-0.02) -- (11.70,-0.02) -- (11.64,-0.02) -- (11.59,-0.02) -- (11.54,-0.03) -- (11.49,-0.03) -- (11.44,-0.03) -- (11.39,-0.03) -- (11.34,-0.03) -- (11.29,-0.03) -- (11.24,-0.02) -- (11.19,-0.02) -- (11.15,-0.02) -- (11.10,-0.01) -- (11.05,-0.01) -- (11.00,0.00);
\draw[] (11.00,0.00) -- (10.90,0.02) -- (10.80,0.04) -- (10.70,0.07) -- (10.60,0.10) -- (10.50,0.14) -- (10.40,0.18) -- (10.30,0.22) -- (10.21,0.27) -- (10.11,0.32) -- (10.02,0.37) -- (9.92,0.43) -- (9.83,0.49) -- (9.74,0.55) -- (9.65,0.61) -- (9.56,0.67) -- (9.46,0.73) -- (9.37,0.80) -- (9.28,0.87) -- (9.19,0.93) -- (9.10,1.00);
\filldraw[,fill=white] (2.00,1.50) rectangle (4.00,2.50);
\filldraw[,fill=white] (6.50,-0.50) rectangle (7.50,2.50);
\filldraw[,fill=white] (11.50,-0.50) rectangle (12.50,0.50);
\draw (0.00,2.00) node{$\bullet$};
\draw (1.00,2.00) node{$\bullet$};
\draw (3.00,2.00) node{$\times\delta$};
\draw (5.00,2.00) node{$\bullet$};
\draw (7.00,1.00) node{$+$};
\draw (9.00,1.00) node{$\bullet$};
\draw (15.00,1.00) node{$\bullet$};
\draw (5.00,0.00) node{$\bullet$};
\draw (12.00,0.00) node{$\iota$};
\end{tikzpicture}
\end{equation}

\noindent
We write~$\varphi:\IT\to\IT$ for the semantics of the net. By
Proposition~\ref{prop:integral}, we have immediately:
\begin{proposition}
  For any function~$s:\R^+\to\R$ in~$\CCT$, $T(\varphi(S(s))) =
  x\mapsto\int_0^xs(t)\d t$.
\end{proposition}
\begin{proof}
  The semantics~$\varphi$ of the net is computed by a fixpoint as explained in
  Section~\ref{sec:nets}, defined by~$\varphi(S(s))(0)=\delta S(s)(0)$ and
  $\varphi(S(s))(\hr{n}+1)= \varphi(S(s))(\hr{n})+\delta
  S(s)(\hr{n}+1)$. Therefore we have
  $\varphi(S(s))(\hr{n})=\delta\sum_{\hr{k}=0}^{\hr{n}} S(s)(\hr{k})$. Finally,
  by Proposition~\ref{prop:integral}, it can be shown that if $\st(\hr{x}\delta) = x\in\Rplus$,
  then $\delta\sum_{\hr{k} = 0}^{\hr{n}} S(s)(\hr{k})\approx \int_0^{x} s(t)\d t$ and
  thus $T(\varphi(S(s)))(x) = \int_0^{x} s(t)\d t$.
\end{proof}

\noindent
This construction can be generalized in order to describe solvers for ordinary
differential equations~\cite{bliudze2009modelling,
  benveniste2010fundamentals}. It should be noticed that the above propositions
show that the choice of the infinitesimal sampling period~$\delta$ is
essentially irrelevant.

Most of the preceding results can be adapted to the case where the streams
considered are only piecewise continuous, with a finite number of
discontinuities. In particular, for any such stream~$s$ we have $T\circ
S(s)\hat=s$ where~$\hat =$ denotes the equality almost everywhere (this
weakening of equality is necessary because of phenomena such as the one
described in the remark following Proposition~\ref{prop:retract}). However, the
formalization of this is obscured by the fact that piecewise continuous
functions, with a finite number of discontinuities, do not form a cpo because
the supremum in~$\CT$ of a directed set of such functions might have an infinite
number of points of discontinuity: this is sometimes referred to as the
\emph{Zeno phenomenon} in the study of hybrid systems.



\section{Conclusion and future works}
\label{sec:concl}
We have defined nets which provide a formal syntax for process networks, studied
the categorical structure of their models, and constructed the
infinitesimal-time model as an internal cpo. The fascinating links between
denotational semantics of concurrent systems and non-standard analysis have
started to be explored only recently and many structures are still yet to be
clarified. As explained above, the study of the infinitesimal-time domain has to
be refined in order to cope with streams which are not necessarily continuous,
and thus to properly model full-fledged hybrid systems. In particular,
Proposition~\ref{prop:retract} fails to be true if we replace~$\CCT$ by~$\CT$:
we plan to investigate generalizations of this property where~$S$ and~$T$ form
an adjunction. We are also investigating possible adaptations of nets in the
case where we consider a synchronous semantics (in which there is a notion of
simultaneity of events). In this setting, the usual delay operator can elegantly
be modeled in feedback categories~\cite{katis1999algebra} (which are traced
monoidal categories with the yanking axiom removed) and we plan to study nets
for those categories. Finally, we envisage many connections with other areas of
denotational semantics. For instance, the trace semantics of Kahn networks is
closely related to game semantics~\cite{mellies2007asynchronous} and it is thus
natural to wonder if non-standard analysis can provide insights about a possible
definition of a ``continuous game semantics'' in the spirit of the geometric
models for concurrent programs~\cite{goubault2000geometry}.

Acknowledgments: the authors would like to warmly thank Simon Bliudze,
Paul-André Melliès, Michael Mislove and Marc Pouzet for all the discussions they
had, directly or indirectly related to the subject of this article, as well as
the anonymous referees for helpful remarks.


\bibliographystyle{plain}
\bibliography{papers}
\end{document}